%% file: AL-path.tex
\documentclass[a4paper,USenglish,cleveref,numberwithinsect,autoref,thm-restate]{lipics-v2021}

\input{mypackages.tex}
\input{preamble.tex}

\title{On the Parameterized Tractability of Packing Vertex-Disjoint {\sf A}-Paths with Length Constraints\thanks{A preliminary version of this paper has appeared in proceedings of MFCS-2024 \cite{BandopadhyayBMS24}}} 

\titlerunning{Parameterized Tractability of {\sf A}-Path Packing} 

\author{Susobhan Bandopadhyay}{Tata Institute of Fundamental Research, Mumbai, India}{susobhan.bandopadhyay@tifr.res.in}{https://orcid.org/0000-0003-1073-2718}{}
 \author{Aritra Banik}{National Institute of Science, Education and Research, An OCC of Homi Bhabha National Institute, Bhubaneswar 752050, Odisha, India}{aritra@niser.ac.in}{}{}
 \author{Diptapriyo Majumdar}{Indraprastha Institute of Information Technology Delhi, New Delhi, India}{diptapriyo@iiitd.ac.in}
 {https://orcid.org/0000-0003-2677-4648}{Supported by Science and Engineering Research Board (SERB) grant SRG/2023/001592.}
 \author{Abhishek Sahu}{National Institute of Science, Education and Research, An OCC of Homi Bhabha National Institute, Bhubaneswar 752050, Odisha, India}{abhisheksahu@niser.ac.in }{}{}

\authorrunning{Bandopadhyay et al.}

\Copyright{Bandopadhyay et al.}

\ccsdesc[500]{Theory of computation~Fixed parameter tractability}

\keywords{Parameterized complexity, $(A,\ell)$-Path Packing, Kernelization, Randomized-Exponential Time Hypothesis, Graph Classes, New Expansion Lemma} 

\relatedversion{} 

 \acknowledgements{We express our sincere gratitude to Professor Karthik C. S. and Professor Saket Saurabh for their invaluable suggestions and enlightening discussions.}

\nolinenumbers 

\EventEditors{John Q. Open and Joan R. Access}
\EventNoEds{2}
\EventLongTitle{42nd Conference on Very Important Topics (CVIT 2016)}
\EventShortTitle{CVIT 2016}
\EventAcronym{CVIT}
\EventYear{2016}
\EventDate{December 24--27, 2016}
\EventLocation{Little Whinging, United Kingdom}
\EventLogo{}
\SeriesVolume{42}
\ArticleNo{23}

\begin{document}

\maketitle



\input{abstract.tex}

\newpage

\tableofcontents

\newpage
\input{intro-new}
 \input{prelims}
\input{hardness}


\input{CVD+A}
\input{CVD+l-Aritra}

\input{vertex-cover}

\input{conclusion}

\input{bbl-references}

\end{document}

%% file: mypackages.tex
\usepackage{graphicx}
\usepackage{tcolorbox}
\usepackage[ruled,vlined,linesnumbered]{algorithm2e}
\usepackage{pdfpages}
\usepackage{amsfonts}
\usepackage{mathtools}
\usepackage{mathrsfs}
\usepackage{xspace}
\usepackage{multicol}
\usepackage{amsmath}
\usepackage{amsthm}
\usepackage{comment}
\usepackage{xcolor}
\usepackage{pdfpages}
\usepackage{amssymb}
\usepackage{array}
\usepackage{comment}
\usepackage{calc}
\usepackage{thm-restate}
\usepackage{cleveref}
\usepackage{thmtools}
\newtheorem{redrule}{\bf Reduction Rule}[section]

\usepackage{dsfont}

\usepackage{tikz,lipsum,lmodern}

\tcbuselibrary{breakable}
\tcbuselibrary{skins}
\PassOptionsToPackage{dvipsnames}{xcolor}

\usepackage{tcolorbox}
\usepackage{framed}
\usepackage{xspace}
\usepackage{optidef}
\usepackage{todonotes}

\tcbuselibrary{breakable}
\tcbuselibrary{skins}
\usepackage{hyperref}
\usepackage{dsfont}
\usepackage{tikz,lipsum,lmodern}
\usepackage{graphicx}
\usepackage{pdfpages}
\usepackage{calc}
\usepackage{mathrsfs}
\usepackage{array}
\usepackage{libertine}
\usepackage{cases}
\usepackage{mathtools}
\usepackage{mdframed}

%% file: preamble.tex
\newcommand{\tw}{{\normalfont \textsf{tw}}}

\newcommand{\vc}{\normalfont \textsf{vc}}
\newcommand{\cvd}{\normalfont \textsf{cvd}}

\newcommand{\dtp}{\normalfont \textsf{dtp}}

\newcommand{\ccP}{{\mathcal{P}}}

\newcommand{\X}{{\mathcal{X}}}

\newcommand{\cR}{\mathcal{R}}

\newcommand{\cQ}{\mathcal{Q}}

\newcommand{\WOH}{\textsf{W[1]}-hard\xspace}

\newcommand{\pnph}{Para-\textsf{NP}-hard\xspace}

\newcommand{\woh}{\textsf{W[1]}\textrm{-hard}\xspace}
\newcommand{\woc}{\textsf{W[1]}\textrm{-complete}\xspace}

\newcommand{\PNPH}{\textsf{Para-NP}\textrm{-hard}\xspace}

\newcommand{\AAA}{{\mathcal A}}

\newcommand{\JJ}{{\mathcal J}}

\newcommand{\MM}{{\mathcal M}}

\newcommand{\shortversion}[1]{}

\newcommand{\OO}{{\mathcal O}}

\newcommand{\rETH}{r\textsf{ETH}\xspace}

\newcommand{\nn}{{\mathbb N}}
\newcommand{\rr}{{\mathbb R}}

\newcolumntype{L}[1]{>{\raggedright\let\newline\\\arraybackslash\hspace{0pt}}m{#1}}
\newcolumntype{C}[1]{>{\centering\let\newline\\\arraybackslash\hspace{0pt}}m{#1}}
\newcolumntype{R}[1]{>{\raggedleft\let\newline\\\arraybackslash\hspace{0pt}}m{#1}}

\numberwithin{equation}{section}

\newcommand{\M}{\mathcal{M}}
\newcommand{\F}{\mathcal{F}}

\newcommand{\I}{\mathcal{I}}
\newcommand{\J}{\mathcal{J}}

\renewcommand{\P}{\mathcal{P}}

	\Crefname{observation}{Observation}{Observations}
	\Crefname{claim}{Claim}{Claims}
	\Crefname{subsection}{Subsection}{Subsections}
	\Crefname{figure}{Figure}{Figures}
	
	\newtheoremstyle{mystyle}
	{\topsep}
	{\topsep}
	{\itshape}
	{}
	{\itshape}
	{.}
	{ }
	{\thmname{#1}\thmnumber{ #2}\thmnote{ (#3)}}%

	\theoremstyle{mystyle}

	\Crefname{myclaim}{Claim}{Claims}

	\definecolor{mygray}{gray}{0.6}
	\definecolor{lightblue}{RGB}{245, 251, 255}

	\newcommand{\ALP}{{\sc $(A, \ell)$-Path Packing}\xspace}
	\newcommand{\alpp}{{\sf ALPP}\xspace}

	\newcommand{\fpt}{\textsf{FPT}\xspace}
	\newcommand{\yes}{{\sf YES}\xspace}
	
	\newcommand{\Pe}{\mathcal{P}\xspace}
	\newcommand{\Ce}{\mathcal{C}\xspace}

	\newcommand{\Pb}{\mathbb{P}}
	\newcommand{\Qq}{\mathcal{Q}}
	\newcommand{\PM}{\mathcal{P}^{\textsc{M}}}
	\newcommand{\PQ}{\mathcal{P}^{\textsc{Q}}}

        \newcommand{\pw}{{\sf pw}\xspace}

	\newcommand{\defproblem}[3]{
		
		\begin{tcolorbox}[colback=gray!5!white,colframe=gray!75!black]
  \vspace{-5pt}
				\begin{tabular*}{\textwidth}{@{\extracolsep{\fill}}lr} #1   \\ \end{tabular*}
				{\bf{Input:}} #2  \\
				{\bf{Question:}} #3
    \vspace{-5pt}
			\end{tcolorbox}
		}

\newenvironment{tightcenter}
{\parskip=0pt\par\nopagebreak\centering}
{\par\noindent\ignorespacesafterend}

\usepackage{tikz}
\usetikzlibrary{calc}
\usepackage{xargs}
\usepackage{xifthen}
\usepackage{framed}

\usepackage{ctable}

\newlength{\RoundedBoxWidth}
\newsavebox{\GrayRoundedBox}
\newenvironment{GrayBox}[1]%
{\setlength{\RoundedBoxWidth}{\textwidth-4.5ex}
	\def\boxheading{#1}
	\begin{lrbox}{\GrayRoundedBox}
		\begin{minipage}{\RoundedBoxWidth}%
		}{%
		\end{minipage}
	\end{lrbox}%
	\begin{tightcenter}%
		\begin{tikzpicture}%
			\node(Text)[draw=black!60,fill=white,rounded corners,%
			inner sep=2ex,text width=\RoundedBoxWidth]%
			{\usebox{\GrayRoundedBox}};
			\coordinate(x) at (current bounding box.north west);
			\node [draw=white,rectangle,inner sep=3pt,anchor=north west,fill=white] 
			at ($(x)+(6pt,.75em)$) {\boxheading};
		\end{tikzpicture}
	\end{tightcenter}\vspace{0pt}%
	\ignorespacesafterend
}    

\newenvironment{problem-box}[2][]{\noindent\ignorespaces%
	\FrameSep=6pt%
	\parindent=0pt%
	\vspace*{-.5em}
	\ifthenelse{\isempty{#1}}{%
		\begin{GrayBox}{\textsc{#2}}%
		}{%
			}

			\begin{tabular*}{\textwidth}{@{\hspace{.1em}} >{\itshape} p{1.2cm} p{0.85\textwidth} @{}}%
			}{
			\end{tabular*}%
		\end{GrayBox}%
		\vspace*{-.5em}
		\ignorespacesafterend
	} 
	
	\newtheorem*{theorem*}{Theorem}

%% file: abstract.tex
\begin{abstract}
Given an undirected graph $G$ and a set $A \subseteq V(G)$, an $A$-path is a path in $G$ that starts and ends at two distinct vertices of $A$ with intermediate vertices in $V(G)\setminus A$. An $A$-path is called an $(A,\ell)$-path if the length of the path is exactly $\ell$.
In the {\ALP} problem (\alpp), we seek to determine whether there exist $k$ vertex-disjoint $(A, \ell)$-paths in $G$ or not.
The problem is already known to be fixed-parameter tractable (FPT) when parameterized by $k+\ell$ using color coding while it remains {\PNPH} when parameterized by $k$ (\textsc{Hamiltonian Path}) or $\ell$ (\textsc{$P_3$-Partition}) alone.
Therefore, a logical direction to pursue this problem is to examine it in relation to structural parameters.
Belmonte et al. [Algorithmica 2022] initiated a study along these lines and proved that {\alpp} parameterized by ${\pw}+|A|$ is {\WOH} where ${\pw}$ is the pathwidth of the input graph. 
In this paper, we strengthen their result and prove that it is unlikely that {\alpp} is fixed-parameter tractable even with respect to a bigger parameter $(|A|+\dtp)$ where {\dtp} denotes the distance from the input graph and a path graph (distance to path).
 Following this, we consider the parameters ${\cvd}+|A|$ and ${\cvd} + \ell$ and provide FPT algorithms with respect to these parameters.
Finally, when parameterized by ${\vc}$, the vertex cover number of the input graph, we provide an $\OO({\vc}^2)$ vertex kernel for {\ALP}. 
\end{abstract}

%% file: intro-new.tex
\section{Introduction}
\label{sec:intro}
{ 
{\sc Disjoint Path} problems form a fundamental class within algorithmic graph theory. These well-studied problems seek the largest collection of vertex-disjoint (or edge-disjoint) paths that satisfy specific additional constraints. Notably, in the absence of such constraints, the problem reduces to the classical {\sc Maximum Matching} problem.
One of the most well-studied variants for disjoint path problems is the {\sc Disjoint $s$-$t$ Path} where given a graph $G$ and two vertices $s$ and $t$, the objective is to find the maximum number of internally vertex-disjoint paths between $s$ and $t$.
This problem is polynomial-time solvable by reducing to {\sc Max-Flow} problem using Menger's Theorem.

Another classical version of the disjoint path problem is {\sc Mader's $\mathcal{S}$-Path}. 
For a graph $G$ and a set $\mathcal{S}$ of disjoint subsets of $V(G)$, an $\mathcal{S}$-Path is a path between two vertices in different members of $\mathcal{S}$. 
Given $G$ and $\mathcal{S}$, the objective of Mader's $\mathcal{S}$-Path problem is to find the maximum number of vertex-disjoint $\mathcal{S}$-paths. 
It is solvable in polynomial time, as demonstrated by Chudnovsky, Cunningham, and Geelen~\cite{Chudnovsky2008AnAF}. One closely related and well-known variant of Mader's $\mathcal{S}$-Path problem is the {\sc $A$-Path Packing} problem~\cite{BelmonteHKKKKLO22, BruhnU22,ChudnovskyGGGLS06,Pap08,HiraiP14}. 
Given a graph $G$ and a subset of vertices $A$, an {\em $A$-path} is a path in $G$ that starts and ends at two distinct vertices of $A$, and the internal vertices of the path are from $V(G) \setminus A$. 
The {\sc $A$-Path Packing} problem aims to find the maximum number of vertex disjoint $A$-paths in $G$. {\sc $A$-Path Packing} problem can be modeled as an $\mathcal{S}$-path problem, where for every vertex $v$, we create a set $\{v\}$ in $\mathcal{S}$.
Consequently, the {\sc $A$-Path Packing} problem becomes polynomial-time solvable.

Recently, Golovach and Thilikos~\cite{MengerTheoremLength} have explored an interesting variant of the classical $s$-$t$ path problem known as \textsc{Bounded $s$-$t$ Path} problem by introducing additional constraints on path lengths. 
In this variant, given a graph $G$, two distinct vertices $s$ and $t$, and an integer $\ell$, one seeks to find the maximum number of vertex disjoint paths between $s$ and $t$ of length at most $\ell$.
Surprisingly, the problem becomes hard with this added constraint in contrast to the classical $s$-$t$ path problem.
In a similar line of study, Belmonte et al.~\cite{BelmonteHKKKKLO22} considered the following variant of the {\sc $A$-Path Packing} problem.

\defproblem{{\ALP} (\alpp)}{An undirected graph $G = (V, E)$, $ A \subseteq V(G)$ and integers $k$ and $\ell$.}{Are there $k$ vertex-disjoint $A$-paths each of length $\ell$ in $G$?}

This version of {\sc $A$-Path Packing} problem is also proved to be intractable~\cite{BelmonteHKKKKLO22}.
While considering this problem in the parameterized framework, the two most natural parameters for {\alpp} are the solution size $k$ and the length constraint $\ell$. 
While parameterized by the combined parameter of $k+\ell$, the problem admits an easy {\fpt} algorithm via {\em color-coding}, parameterized by the individual parameters  $k$ and $\ell$, the problem becomes {\pnph} due to reductions from {\sc Hamiltonian Path} for $k$ and {\sc $P_3$-Packing} for $\ell$. 

Although it may appear that one has exhausted the possibilities for exploration of the problem within the parameterized framework, another set of parameters, known as \emph{structural parameters}, emerges, allowing for further investigation.
Belmonte et al.~\cite{BelmonteHKKKKLO22} initiated this line of study by considering the size of set $A$ $(|A|)$ as a parameter. They proved that the problem is {\woh} parameterized by the combined parameter, pathwidth of the graph and the size of $A$, i.e. ${\pw} +|A|$.
This {\woh} result translates to the {\woh} result when parameterized by $\tw+|A|$, i.e. the combined parameter treewidth of the input graph, and $|A|$.
 
The intractability result for the parameter ${\tw}+|A|$ refutes the possibility of getting {\fpt} algorithms for many well-known structural parameters.
Nonetheless, one of the objectives of structural parameterization is to delimit the border of the tractability of the problem, i.e., determining the smallest parameter for which the problem becomes {\fpt} or the largest parameters that make the problem {\sf W}-hard. 
Therefore one natural direction is to study {\alpp} with respect to parameters that are either larger than or incomparable to ${\tw} + |A|$. 

\subparagraph*{Our Contribution:}
As our first result, we improve upon the hardness result of Belmonte et al. \cite{BelmonteHKKKKLO22} by showing hardness for a much larger parameter, ${\dtp} + |A|$.
Here $\dtp$ denotes the distance from the input graph to a path graph (formal definitions of all the parameters can be found in \Cref{sec:prelim}).
In particular, our algorithmic lower bound result is a W-hardness result that proves something stronger under {\rETH} (Conjecture \ref{conj:rETH}) which we state below.

\begin{restatable}{theorem}{finalresult}
\label{thm:ALP-hardness-result}
Unless Conjecture \ref{conj:rETH} fails, there is no randomized algorithm for {\ALP} that runs in $f({\dtp} +|A|)\cdot n^{o(\sqrt{{\dtp} +|A|})}$-time and correctly decides with probability at least $2/3$ when ${\dtp}$ is the distance to a path graph.
\end{restatable}

We present a randomized reduction from a well-known {\woh} problem, which establishes the hardness of {\ALP} under the assumption of randomized Exponential Time Hypothesis (\rETH). 
The randomized reduction technique employed in our proof is highly adaptable and can be utilized to demonstrate the hardness of various analogous problems. 
We use the following lemma to prove our hardness result, which can be of independent interest.

\begin{restatable}[\rm Separation Lemma]{lemma}{generalizedisolation}
\label{lemma:generalized-isolation}
Let $(X,\mathcal{F})$ be a set system where $\mathcal{F}$ is a family of subsets of $X$. For an arbitrary assignment of weights $w:X \mapsto [M]$, let $w(S)=\sum_{x\in S} w(x)$ denote the weight of the subset $S\subseteq X$. For any random assignment of weights to the elements of $X$ independently and uniformly from $[M]$, with probability at least $1-\frac{{|\mathcal{F}| \choose 2}}{M}$, each set $S\in \mathcal{F}$ has a unique weight.
\end{restatable}

In addition to refining the boundaries of hardness, we have also considered the problem with respect to the parameter of $\cvd$ denoting the cluster vertex deletion number in combination with the natural parameters $|A|$ and $\ell$. 
The incomparability of $\cvd$ with $\tw$ and $\pw$ makes it an intriguing parameter to explore.
We have proved that {\alpp} is fixed-parameter tractable with respect to the parameters ${\cvd} + |A|$ as well as ${\cvd} +\ell$.
Formally, we prove the following two fixed-parameter tractability results.

\begin{restatable}{theorem}{cvdplusA}
\label{thm:ALP-cvd-plus-A}
{\ALP} is {\fpt} when parameterized by ${\cvd} + |A|$.
\end{restatable}

\begin{restatable}{theorem}{cvdlShortALPPresult}
\label{thm:cvd+l-Short-ALPP-result}
{\ALP} is {\fpt} when parameterized by ${\cvd} + \ell$.
\end{restatable}

Finally, for {\ALP}, we design a quadratic vertex kernel for a larger parameter vertex cover number ${\vc}$ (larger than each of ${\cvd}, {\pw}$ and ${\tw}$) using new $q$-expansion lemma and carefully problem specific approach.
 Formally, we prove the following result.
 
\begin{restatable}{theorem}{vcResult}
\label{thm:vc-ALPP-result}
{\ALP} parameterized by ${\vc}$, the vertex cover number admits a kernel with	$\OO({\vc}^2)$ vertices.
\end{restatable}

 
\begin{figure}[h!]
	\centering
\includegraphics[width=.7\textwidth]{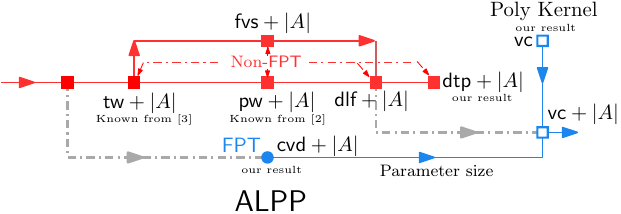}
	\caption{Structural Parameterizations of {\alpp}.  The arrow represents the hierarchy of different structural parameters, while the dashed line represents the parameters that have yet to be explored in the context of our problems.  }
	\label{fig:hierarchy-of-parameters}
\end{figure}

We organize our paper as follows.
In Section \ref{sec:prelim}, we introduce the basic definitions of graph theory and parameterized complexity.
Next, in Section \ref{sec:hardness-result}, we prove our main result where we prove Theorem \ref{theorem:rETHIS}.
After that, we prove Theorems \ref{thm:ALP-cvd-plus-A} and Theorem \ref{thm:cvd+l-Short-ALPP-result} in Section \ref{sec:cvd-with-A} and Section \ref{sec:cvd+l} respectively.
Finally, in Section \ref{sec:ALPP-vc}, we prove Theorem \ref{thm:vc-ALPP-result}. 

%% file: prelims.tex
\section{Preliminaries and notations}
\label{sec:prelim}

\subparagraph*{Sets, numbers and graph theory:}
We use $\nn$ to denote the set of all natural numbers and $[r]$ to denote the set $\{1,\ldots,r\}$ for every $r \in \nn$.
Given a finite set $S$ and $r \in \nn$, we use ${{S}\choose{r}}$ and ${{S}\choose{\leq r}}$ to denote the collection of subsets of $S$ with exactly $r$ elements and at most $r$ elements respectively.
We use standard graph theoretic notations from the book by Diestel \cite{Diestel-Book}.
For a graph $G$, a set of vertices $W \subseteq V(G)$ is said to be {\em independent} if for any pair of vertices $u, v \in W$, $uv \notin E(G)$.
Equivalently, a set $W \subseteq V(G)$ is independent if every pair of vertices of $W$ is nonadjacent to each other.
For any two vertices $x,y$ and a path $P$, we denote $V[x,y]$ as the number of vertices in the subpath between $x$ and $y$ in $P$. And for a path $P$ we denote the set of vertices in $P$ by $V(P)$. Further, for a collection $\Pe$ of paths, $V(\Pe)=\{\bigcup_{P_i\in \Pe} V(P_i)\}$. 
In a graph $G$, let $P_i(v_1,v_2,\cdots, v_j)$ be a path and $X\subseteq V$ be a set. An {\em ordered intersection} of $P_i$ with $X$, denoted as $X_i=(v_a, v_b, \cdots, v_p)$, is defined as $V(P_i)\cap X=X_i$, where the ordering of the vertices in $X_i$ is the same as that in $P_i$.
Additionally, we define ${X_1, X_2, \cdots, X_x}$ as an {\em ordered partition} of $X$ if each $X_i$ is an ordered set and $\bigcup_{i=1}^{x}X_i=X$.
Given a path $P_i(v_1,v_2,\cdots, v_j)$, we denote $P_i(v_1,v_2,\cdots, v_{j-1})$ as $P_i\setminus \{v_{j}\})$.
{Given two vertex-disjoint paths $P_1$ and $P_2$ such that the last vertex of $P_1$ is adjacent to the first vertex of $P_2$. 
We use $P_1 \cdot P_2$ to denote the {\em concatenation} of path $P_1$ with the path $P_2$.
Informally speaking, in $P_1 \cdot P_2$, the path starts with the vertices of $P_1$, followed by the vertices of $P_2$.}
A graph $G$ is {\em bipartite} if its vertex set $V(G)$ can be partitioned into $A \uplus B$ such that for every $uv \in E(G)$, $u \in A$ if and only if $v \in B$.
We often use $(A, B)$ to denote the {\em bipartition} of $G$.
A {\em $2$-interval} $\mathcal{I}_i$ is a disjoint pair of intervals $\{I_i^a,I_i^b\}$ on a real line.
We say that a pair of $2$-intervals, $\mathcal{I}_i$ and $\mathcal{I}_j$ {\em intersect} if they have at least one point in common, that is $\{I_i^a\cup I_i^b\}\cap\{I_j^a\cup I_j^b\}\neq \emptyset$.
Conversely, if two $2$-intervals do not intersect, they are called {\em disjoint}.

A {\em $2$-interval representation} of a graph $G$ is a set of two intervals $\J$ such that there is a one to one correspondence between $\J$ and $V(G)$ such that there exists an edge between $u$ and $v$ if and only if the $2$-intervals corresponding to $u$ and $v$ intersect. 
A graph is a {\em $2$-interval graph} if there is a {\em $2$-interval representation} for $G$.

\subparagraph*{New Expansion Lemma:}
Let $q$ be a positive integer and $H$ be a bipartite graph with bipartition $(A, B)$.
For $\widehat A \subseteq A$ and $\widehat B \subseteq B$, a set $M \subseteq E(H)$ is a {\em $q$-expansion} of $\widehat A$ onto $\widehat B$ if 
\begin{enumerate}[(i)]
	\item every edge of $M$ has one endpoint in $\widehat A$ and the other endpoint in $\widehat B$,
	\item every vertex of $\widehat A$ is incident to exactly $q$ edges of $M$, and
	\item exactly $q|\widehat A|$ vertices of $\widehat B$ are incident to some edges of $M$.
\end{enumerate}

A vertex of $\widehat A \cup \widehat B$ is {\em saturated} by $M$ if it is an endpoint of some edge of $M$, otherwise it is {\em unsaturated}.
By definition, all vertices of $\widehat A$ are saturated by $M$, but $\widehat B$ may contain some unsaturated vertices.
Observe that a $1$-expansion of $\widehat A$ onto $\widehat B$ is a matching of $\widehat A$ onto $\widehat B$ that saturates all vertices of $\widehat A$.
We use the following proposition to prove one of our results.

\begin{proposition}[New $q$-Expansion Lemma, \cite{Babu0R22,FominLLSTZ19,JacobMR23}]
\label{lemma:new-expansion-lemma}
Let $H = (A \uplus B, E)$ be a bipartite graph and $q > 0$ be a positive integer. 
Then there exists a polynomial-time algorithm that computes $\widehat A \subseteq A$, $\widehat B \subseteq B$, and $M \subseteq E(H)$ such that 
\begin{itemize}
	\item $M$ is a $q$-expansion from $\widehat A$ onto $\widehat B$,
	\item $N_G(\widehat B) \subseteq \widehat A$, and
	\item $|B \setminus \widehat B| \leq q|A \setminus \widehat A|$.
\end{itemize}
\end{proposition}

Note that we do not need a precondition that $|B| > q|A|$.
Moreover, it is possible that $\widehat A, \widehat B =  \emptyset$.
But, if $|B| > q|A|$, then it follows due to the third property of above proposition that $\widehat B \neq \emptyset$ and $|\widehat B| > q|\widehat A|$ even if $\widehat A = \emptyset$.
The idea is to delete some vertex $u$ from $\hat B$ that is not spanned by $M$.

\subparagraph*{Graph parameters:} Given a graph $G$, the {\em distance to path graph}, denoted by ${\dtp}(G)$ is the minimum number of vertices to be deleted from $G$ to get a path.
The {\em distance to cluster graph}, denoted by ${\cvd}(G)$ is the minimum number of vertices to be deleted to get a disjoint union of cliques.
A set of vertices $S \subseteq V(G)$ is called the {\em vertex cover} of $G$ if for every $uv \in E(G)$, $u \in S$ or $v \in S$.
The {\em vertex cover number} of $G$, denoted by ${\vc}(G)$ is the size of a vertex cover of $G$ with minimum number of vertices.
For each of the parameters, when the graph is clear from the context, we drop the mentioning of $G$.

\subparagraph*{Parameterized complexity and kernelization:}
 A {\em parameterized problem} is a language $L \subseteq \Sigma^* \times \nn$
 where $\Sigma$ is a finite alphabet.
 The input instance to a parameterized problem is $(x, k) \in \Sigma^* \times \nn$.
 This $k$ is called the parameter.
 
 \begin{definition}[Fixed-Parameter Tractability]
 \label{defn:FPT}
 A parameterized problem $L$ is said to be {\em fixed-parameter tractable} (or {\em FPT}) if given an input instance $(x, k) \in  \Sigma^* \times \nn$ to a parameterized problem $L$, there exists an algorithm $\AAA$ that runs in $f(k)\cdot |x|^{\OO(1)}$ time and correctly decides $L$.
 \end{definition}

 The algorithm $\AAA$ illustrated in the above definition running in $f(k)\cdot |x|^{\OO(1)}$-time is called a {\em fixed-parameter algorithm} (or an {\em FPT algorithm}).
A {\em kernelization} of a parameterized problem is a polynomial-time preprocessing algorithm that outputs an equivalent instance of a problem of `small size'.
The formal definition is as follows.

\begin{definition}[Kernelization]
\label{defn:kernel}
A parameterized problem $L \subseteq \Sigma^* \times \nn$ is said to admit a {\em kernelization} if given $(I, k)$, there is a preprocessing algorithm that runs in polynomial-time and outputs an equivalent instance $(I', k')$ such that $|I'| + k' \leq g(k)$ for some computable function $g: \nn \rightarrow \nn$.
\end{definition}

The output instance $(I', k')$ satisfying $|I'| + k' \leq g(k)$ is called the {\em kernel} and the function $g(\cdot)$ is the {\em size} of the kernel.
If the size of the kernel is a polynomial function, then $L \subseteq \Sigma^* \times \nn$ is said to admit a {\em polynomial kernel}.
The process of transforming the input instance into a (polynomial) kernel is commonly referred to as \textit{(polynomial) kernelization}.
It is well-known that a decidable parameterized problem is FPT if and only if admits a kernelization \cite{CyganFKLMPPS15,DowneyF13}.

Since one of our hardness result uses the following conjecture {\rETH} (Randomized Exponential-Time Hypothesis) introduced by Dell et al. \cite{DellHMTW14}, we formally state it below.

\begin{conjecture}[{\rETH}]
\label{conj:rETH}
There exists a constant $c > 0$, such that no randomized algorithm can solve {\sc $3$-SAT} in time $\OO^*(2^{cn})$ \footnotetext{$\OO^*$ notation suppresses the polynomial factors.} with a (two-sided) error probability of at most $\frac{1}{3}$, where $n$ represents the number of variables in the $3$-SAT instance.
\end{conjecture}

%% file: hardness.tex

\section{Algorithmic lower bound based on {\rETH} when parameterized by ${\dtp} + |A|$}
\label{sec:hardness-result}

This section establishes that {\alpp} is unlikely to be fixed-parameter tractable with respect to the combined parameter $({\dtp}+|A|)$ under standard complexity theoretic assumptions. We begin by presenting some key ideas that will be instrumental in our subsequent hardness reduction.

 
\subsection{Essential results}
We show the hardness for {\alpp} under the assumption that Conjecture \ref{conj:rETH} holds.
In the realm of parameterized complexity, {\rETH} has been widely employed to prove hardness for many well-known parameterized problems.
The input instance to the $k$-\textsc{Clique} is an undirected graph $G$ and the question is whether $G$ has a set of at least $k$ vertices that induces a complete subgraph.
The following proposition about $k$-\textsc{Clique} can be derived from  Theorem 4.1 in \cite{ChalermsookCKLM20} when $\varepsilon = 1/m$ in Conjecture 2.5 of \cite{ChalermsookCKLM20}.
 
\begin{restatable}{proposition}{rrethkclique}
\label{prop:r-ETH-k-CLIQUE}
Unless Conjecture \ref{conj:rETH} fails, there is no randomized algorithm that decides $k$-\textsc{Clique} in time $f(k)\cdot n^{o(k)}$ correctly with probability at least $2/3$.
\end{restatable}
	
Similarly, given a graph $G$, the {\sc $k$-Independent Set} problem asks whether there exists an independent set in $G$ having (at least) $k$ vertices.
We establish the intractability result for {\alpp} parameterized by $(|A|+\dtp(G))$ through a `parameter preserving reduction' from the {\sc $k$-Independent Set} problem on a $2$-interval graph, which in turn was shown to be $\woc$ following a reduction from the $k$-{\sc Clique} problem by Fellows et al. \cite{FellowsHRV09}.
Observe that given a $2$-interval graph $G$ and a $2$-interval representation $\J$ of $G$ a $k$-independent set $W$ for $G$ corresponds to a set of pairwise disjoint $2$-intervals in $\J$.
Fellows et al. \cite{FellowsHRV09} presented a parameterized reduction from an arbitrary instance $(G,k)$ of the $k$-\textsc{Clique} problem to an instance $(\J,k')$ of the {\sc $k'$-Independent Set} problem on a $2$-interval graph such that there exists a $k$ size clique in $G$ if and only if there exists a $k'=k+3{k\choose 2}$ sized independent set in $\J$.
Thus, we have the following proposition.

\begin{restatable}{proposition}{rethIS}\label{theorem:rETHIS}
Unless {\rETH} fails, there is no randomized algorithm that decides {\sc $k$-Independent Set} on a $2$-interval graph in $f(k)\cdot n^{o(\sqrt{k})}$-time correctly with probability at least $2/3$.
\end{restatable}

Next, we present a lemma, which we will use later on. We believe that this can be of independent interest and applicable to various other problem domains. This lemma is similar to the well-known isolation lemma~\cite{Mulmuley1987}. Recall that $[r]$ denotes the set $\{1,\ldots,r\}$ where $r\in \mathbb{N}^+$.


{\generalizedisolation*}

\begin{proof}
    Let $w$ be a random assignment of weights to the elements of $X$ independently and uniformly from $[M]$ and $S_1$ and $S_2$ be two arbitrary sets in $\mathcal{F}$. Our objective is to find the probability of the event ``$w(S_1)=w(S_2)$''. Observe that if $S_1\setminus S_2 =\emptyset $ or $S_2\setminus S_1=\emptyset$, then $\mathds{P} \bigl(w(S_1)=w(S_2)\bigr)=0$. Let $S_1\setminus S_2=\{x_1,x_2, \cdots, x_a\}$ and $S_2\setminus S_1=\{y_1,y_2,\cdots y_b\}$. 
    We define a random variable $W_{12}$ as follows.
    \begin{align*}
         &W_{12}=\{w(x_1)+\cdots + w(x_a)\}-\{w(y_2)+\cdots + w(y_b)\}
         =w(S_1)-w(S_2)+w(y_1)
    \end{align*}
    From the law of total probability, we have the following.
    \begin{align*}
         &\mathds{P} \bigl(w(S_1)=w(S_2)\bigr)= \sum \mathds{P}\bigl(w(S_1)=w(S_2)|W_{12}=z\bigr)\cdot \mathds{P}(W_{12}=z) & \\
         &= \sum_{z\in [M]}\mathds{P}\bigl(w(y_1)=z\bigr)\cdot \mathds{P}(W_{12}=z) & & \llap{\text{ if 
         $z\notin [M]$ then $\mathds{P}\bigl(w(y_1)=z\bigr)=0$
         }}\\
         &= \sum_{z\in [M]} \frac{1}{M} \mathds{P}(W_{12}=z)= \frac{1}{M} \sum_{z\in [M]} \mathds{P}(W_{12}=z)\leq\frac{1}{M}
    \end{align*}
    \noindent
    Using Boole's inequality, we can prove that none of the two sets in $\mathcal{F}$ are of equal weight with probability at least $1-\frac{{|\mathcal{F}| \choose 2}}{M}$. Thus, the claim holds.  
\end{proof}

\subsection{Hardness proof}
We are ready to present a randomized reduction from {\sc $k$-Independent Set} in $2$-interval graphs to the {\alpp}.
Throughout this section, we denote the family of path graphs by $\Gamma$.
Let $(G_{\J},k)$ be an instance of $k$-{\sc Independent Set} problem in $2$-interval graph where the set of $2$-intervals  representing $G_{\J}$ be $\J$. 
Now, we present a randomized construction from $(G_{\J},k)$ to an instance $(G,A,M,\ell,k)$ of {\alpp}.
In this reduction, $H= G - M$ is in $\Gamma$ and $|M|= 4k$.
We assume that we are given $\J$. The construction of $G$ is done in two phases.
In the first phase, we generate a set of points $P$ on the real line $\rr$.
In the second phase, we construct the graph $G$.
Observe that the points in $P$ naturally induce a path graph $H$ which is defined as follows.
Corresponding to each point in $P$, we define a vertex in $V(H)$, and there is an edge between two vertices if the points corresponding to them are adjacent in $\mathbb{R}$. We additionally add $4k$ vertices.
Before detailing our construction, let us establish a few notations and assumptions that can be accommodated without changing the combinatorial structure  of the problem. 
Let $\J=\{\mathcal{I}_1,\mathcal{I}_2,\cdots,\mathcal{I}_n\}$ where each $2$-interval $\I_j$ is a collection of two intervals $I_j^a$ and $I_j^b$. We presume that all the intervals in $\J$ are inside the interval $[0,1]$ in the real line. 
Furthermore, we assume that all the endpoints of the intervals are distinct, and the distance between any two consecutive endpoints is at least $2\varepsilon$, where $\varepsilon$ is an arbitrarily small constant.
We use $L(I)$ and $R(I)$ to denote the left and right endpoints of an interval $I$, respectively.

\begin{tcolorbox}[bicolor,
  colback=cyan!5,colframe=cyan!5,boxrule=0pt,frame hidden]
\subparagraph*{Construction Phase $1$.}
We place a set of points  $P$ on $\rr$ as follows. For each interval $I_j^c$, where $c\in\{a,b\}$ and $j\in [n]$, we generate two random numbers $n_r(L(I_j^c))$ and $n_r(R(I_j^c))$ between $1$ and $N$. We will decide on the value of $N$ at a later stage. 
\vspace{3pt}

Let $P(L(I_j^c))$ be the set of $n_r(L(I_j^c))$ equally spaced points in the interval $[L(I_j^c),L(I_j^c)+\varepsilon]$.
The first point of $P(L(I_j^c))$ coincides with $L(I_j^c)$, and the last point coincides with $L(I_j^c)+\varepsilon$.
Let  $P_L=\bigcup_{c \in \{a, b\}, j \in [n]} P(L(I_j^c))$.
Similarly, let $P(R(I_j^c))$ be the set of $n_r(R(I_j^c))$ equally spaced points in the interval $[R(I_j^c)-\varepsilon,R(I_j^c)]$.
The first point of $P(R(I_j^c))$ coincides with $R(I_j^c)-\varepsilon$ and the last point coincides with $R(I_j^c)$.
Let $P_R = \bigcup_{c\in\{a,b\}, j \in [n]} P(R(I_j^c))$.
An illustration of this process can be seen in Figure \ref{fig:al3}.
Observe that the cardinality of $P_L \cup P_R$ is at most $4nN$.

\vspace{3pt}
Consider ${\mathcal{I}_j=(I_j^a,I_j^b)}$ be a $2$-interval, and  $n_j$ be the total number of points from $P_L\cup P_R$ that are inside $\mathcal{I}_j$.
Let $\overline{n_j}= 8nN - n_j$ for all $j \in [n]$. Define $C_j$ as the collection of $\overline{n_j}$ points evenly distributed within the interval $[2j, 2j+1]$ and let $P_C=\cup_{j \in [n]} C_j$.
To ease the notations, we denote $L_j=2j$, $R_j=2j+1$ and $I_j$ as the interval $(L_j,R_j)$. 
Observe that the total number of points inside $I_j^a,I_j^b$ and $I_j$ is $8nN$.
\vspace{3pt}

We add a large number of points (exactly $8nN+4$ many) between $R_j$ and $L_{j+1}$ for $ j \in [n]$ and denote all these points by $P_X$. 
Let $P=P_L \cup P_R\cup P_C \cup P_X$. 
\end{tcolorbox}

\vspace{-10pt}
\begin{figure}[h!]
		\centering
		\includegraphics[width=1\textwidth]{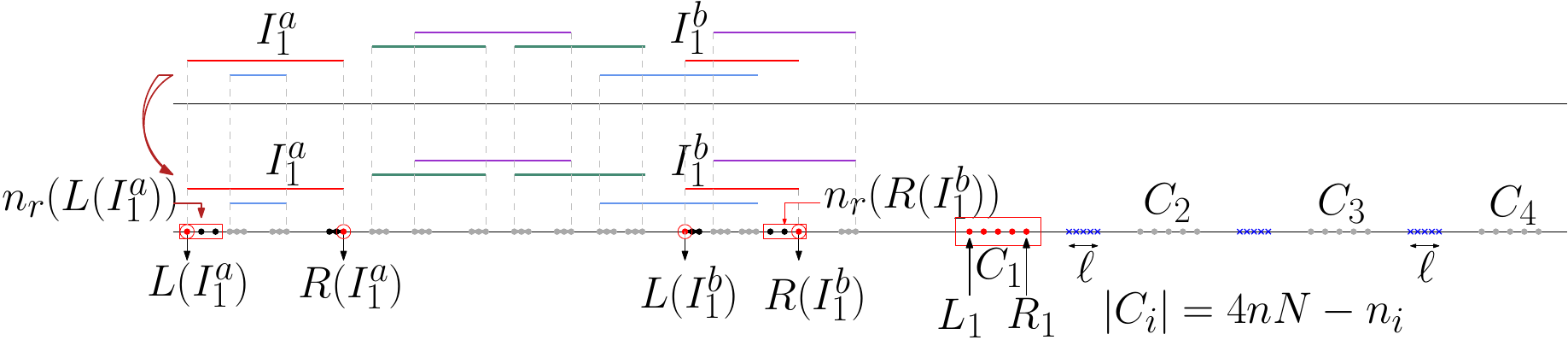}
		\caption{Illustration of Construction Phase 1. Note that $|C_i| = \overline{n_i} = 8nN - n_i$.}
		\label{fig:al3}
	\end{figure}

\begin{tcolorbox}[bicolor, colback=cyan!5,colframe=cyan!5,boxrule=0pt,frame hidden]

\subparagraph{Construction Phase $2$.}
Consider the set of points $P$ in $\rr$. Observe that there is a natural ordering among the points in $P$. We define adjacency based on this ordering.
Consider the path graph $G_P$ induced by $P$. Specifically, we introduce a vertex for each point in $P$ and add an edge between two vertices if and only if the points corresponding to them are adjacent. With slight abuse of notation, we denote the vertex corresponding to a point $p$ by $p$. For any interval $I$, let $V(I)$ denote the number of points/vertices within the interval $I$.
For two points $a$ and $b$ in $P$, $\lambda[a,b]$ denotes the path from $a$ to $b$ in $G_P$. 
\vspace{3pt}

We construct $G$, by setting $V(G)=V(G_P)\cup V_M$ where $V_M=\{a_i,b_i,c_i,d_i|~ i\in [k]\}$.
And, $E(G)=E(G_P) \cup E_a \cup E_b\cup E_c\cup E_d$, where 
$E_a=\{(a_i,L(I_j^a))\mid j\in [n] \textrm{ and } i\in [k]\}$, 
$E_b=\{(b_i,R(I_j^a)), (b_i,L_j)\mid j\in [n] \textrm{ and } i\in [k]\}$, 
$E_c=\{(c_i,L(I_j^b)), (c_i,R_j)\mid j\in [n] \textrm{ and } i\in [k]\}$,
$E_d=\{(d_i,R(I_j^b))\mid j\in [n] \textrm{ and } i\in [k]\}$.

Informally, the edges are defined as follows. Each $a_i\in V_M$ is adjacent to $L(I_j^a)$ for all $j\in [n]$ (see Figure \ref{fig:al5}). Similarly each $b_i\in V_M$ is adjacent to $R(I_j^a)$ for all $j\in [n]$. Additionally, each $b_i$ is adjacent to every other $L_j$ for all $j\in [n]$. Each $c_i$ is adjacent to $L(I_j^b)$ and $R_j$, and each $d_i$ is adjacent to $R(I_j^b)$.

We denote $A=\{a_i,d_i|~i\in [k]\}$ and set $\ell=8nN+4$.
\noindent 	
\end{tcolorbox}

\begin{figure}[h!]
		\centering
		\includegraphics[width=0.8\textwidth]{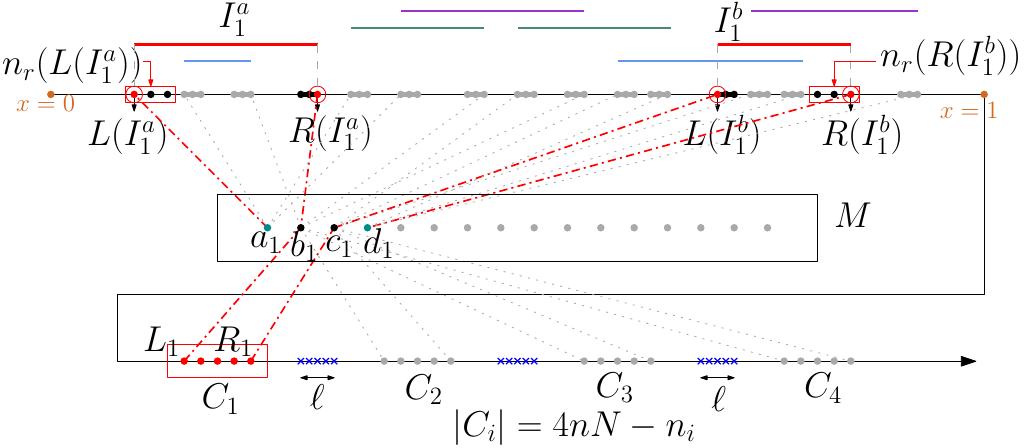}
		\caption{Illustration of construction of {\ALP}. Note that $|C_i| = \overline{n_i} = 8nN - n_i$.}
		\label{fig:al5}
	\end{figure}

Next, we demonstrate that if there exists a $k$-size independent set in $G_{\J}$, then with probability $1$ there exist $k$ many vertex-disjoint $(A, \ell)$-paths in $G$ where $\ell=8nN + 4$ (Lemma~\ref{lemma:easypart}). Following this, we show that if there exist $k$ many vertex-disjoint $(A, \ell)$-paths in $G$ where $\ell=8nN + 4$, then with \emph{high} probability there is a $k$-size independent set in $G_{\J}$.


%
%

\begin{restatable}{lemma}{easypart}
\label{lemma:easypart}
If there are $k$ disjoint $2$-intervals in $\J$, then there are $k$ vertex-disjoint $(A,\ell)$-paths in $G$.
\end{restatable}

\begin{proof}
    Without loss of generality, let us assume that the $k$ disjoint $2$-intervals are $\mathcal{I}_1,\mathcal{I}_2,\ldots,\mathcal{I}_k$. Recall that for two vertices $a$ and $b$ in the path graph $G_P$, $\lambda[a,b]$ denotes the path from $a$ to $b$ in $G_P$. Consider the following set of paths, defined for $1\leq i \leq k$,

 $$\lambda_i=a_i\cdot \lambda[L(I_i^a),R(I_i^a)]\cdot b_i \cdot \lambda[L_i,R_i]\cdot c_i\cdot \lambda[L(I_i^b),R(I_i^b)]\cdot d_i$$
    
    
    Observe that the paths in $\{\lambda_i|1\leq j\leq k\}$ are pairwise vertex-disjoint, each having $\ell$ vertices, with two endpoints at two different vertices in $A$ (See Figure \ref{fig:al4} for an illustration). By construction, they contain $\ell=8nN + 4$ many vertices. Thus, the claim holds.
\end{proof}


	\begin{figure}[h!]
		\centering
		\includegraphics[width=0.8\textwidth]{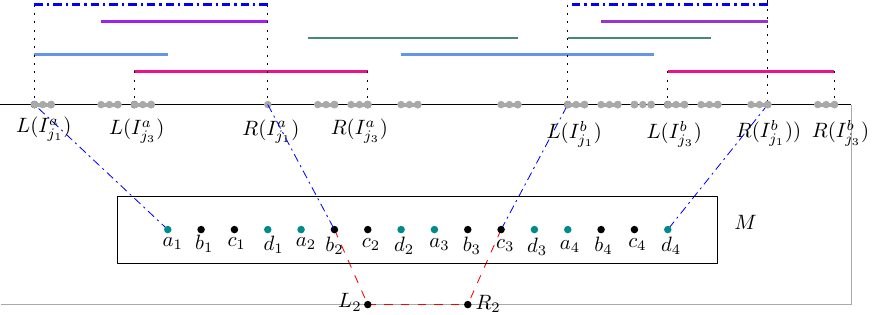}
		\caption{Construction of the paths $\lambda_i$ in the proof of Lemma \ref{lemma:easypart}.}
		\label{fig:al4}
	\end{figure}

Next, we show that if there exist $k$ many vertex-disjoint $(A, \ell)$-paths in $G$ where $\ell=8nN + 4$, then with \emph{high} probability there is a $k$ size independent set in $G_{\J}$.

Recall that the set of intervals $\{I_j|j\in[n]\}$ where $I_j=[L_j,R_j]$. Let $V_j=V(I_j)$. And, the path corresponding to $I_j$ is $\lambda[L_{j},R_{j}]$.  Let us define $\mathcal{C}=\{C_i|1\leq i\leq n\}$. Next, we have the following observation.

\begin{restatable}{observation}{dtponeblock}
\label{obs:dtp-one-block}
For any $A$-path $L$ of length exactly $\ell=8nN+4$ in $G$, there exists an integer $y\in [n]$ such that $V_y\subset V(L)$ and for any $i\neq y$, $V_i\cap V(L)=\emptyset$. Further $L$ contains the subpath $b_{x} \cdotp \lambda[L_{y},R_{y}] \cdotp c_{z}$ for some $x,z\in [k]$.
\label{obs:subpathblock}
\end{restatable}
	
\begin{proof}
As the total number of points/vertices in $[0,1]$ is at most $4nN<\ell$, note that every $(A,\ell)$-path should contain all points/vertices from at least one set of points from $\mathcal{C}$.
Observe that any path with a length of exactly $\ell$ contains exactly one set of points $C_i\in \mathcal{C}$  as between two consecutive set of points in $\mathcal{C}$ there are more than $\ell$ points.
As only neighbors of $C_i$ which is not in $H$ is $b_x$ and $c_z$, where $x,z\in [k]$, therefore, the path must include a subpath of the form $b_{x} \cdot \lambda[L_{y},R_{y}] \cdot c_{z}$ for some $y\in [n]$ and $x,z\in [k]$. 
\end{proof}

Let $\P=\{P_1,\ldots,P_k\}$ be a set of $k$ vertex-disjoint $A$-paths of length exactly $\ell$ in $G$. We show that with \emph{high} probability, there is a $k$ size independent set in $G_{\J}$.
	
\begin{restatable}{observation}{pathstructure}
Each path $P_i\in \P$ is of the form $p\cdotp  \lambda[r,s] \cdotp b_{x} \cdotp \lambda[ L_{y},R_{y}] \cdotp c_{z}\cdotp  \lambda[t,u]  \cdot q$ where $p,q\in A$, $r,s,t,u\in M$, $y\in[n]$ and $x,z\in[k]$.
\label{obs:pathstructure}
\end{restatable}
	
\begin{proof}
Since the total number of points in $A$ is $2k$, each path must contain exactly two points from $A$.
It follows from Observation \ref{obs:subpathblock}  that any such path contains a subpath of the form $b_{x} \cdotp \lambda[L_{y},R_{y} ]\cdotp c_{z}$.
By construction, for any path $P$, that connects $p$ with $b_{x}$, $P\setminus\{p,b_x\}$ induces a continuous set of points on the path graph with endpoints in $M$ (see Figure~\ref{fig:al4}).
A similar argument can be made for $c_z$ and $q$ as well.
Thus, the claim holds.
\end{proof}

{An unordered tuple $I$ is a {\em $3$-interval} if it is a triplet of three pairwise-disjoint intervals in real line.}
Observation~\ref{obs:pathstructure} indicates that from the disjoint paths in $\P$, it is possible to construct a disjoint set of  $3$-intervals $\J=\{((x_i,y_i),(z_i,w_i), I_{j})|1\leq i\leq k\}$ where $x_i,y_i,z_i,w_i\in M$ and $I_j=(L_j,R_j)$.
Next we prove that with \emph{high} probability the interval $(x_i,y_i)=I_{y_i}^a$ and the interval $(z_i,w_i)=I_{y_i}^b$, which will give us the desired set of $k$ disjoint $2$-intervals.

\begin{restatable}{lemma}{hardpart}
\label{lemma:hardpart}
   If there are $k$ vertex-disjoint $(A,\ell)$-paths in $G$, then with probability at least $\frac{2}{3}$, there are $k$ disjoint $2$-intervals in ${\J}$.
\end{restatable}

\begin{proof}
Let $a,b,c,d$ be any four points in $M$. Without loss of generality, assume that $a<b<c<d$. Consider the $2$-interval $((a,b),(c,d))$ defined by $a,b,c,d$. Let $\J_{\F}$ be the set of all such intervals, formally defined as $\J_{\F}=\{((a,b),(c,d))|a,b,c,d\in \textrm{ and } a<b<c<d\}$. Note that $|\J_{\F}|< n^4 $. 
It follows from  Lemma~\ref{lemma:generalized-isolation} that each $2$-interval in $\J_{\F}$ contains a unique number of points with probability at least $1-\frac{{{n^4}\choose{2}}}{N}$.
If we set $N=3{{n^4}\choose{2}}$, then with probability at least $\frac{2}{3}$, every 2-interval in $\J_{\F}$ contains a unique number of points.
Thus with probability at least $\frac{1}{3}$, for every $1\leq j\leq n$, no other 2-interval except $(I_j^a,I_j^b)$ contains $\ell-V(I_j)-4$ many points where $V(I)$ denotes the number of points in the interval $I$.
Therefore with probability at least $\frac{2}{3}$, for every $3$-interval $((x_i,y_i),(z_i,w_i), I_{y_i})$ in $\J$, $(x_i,y_i)=I_{y_i}^a$ and $(z_i,w_i)=I_{y_i}^b$.
\end{proof}

Observe that in the construction of the graph $G$ from $(\JJ, k)$, the vertex set is $V_M$ is such that $G - V_M$ is a path, $A \subseteq V_M$ and $|V_M| \leq 4k$.
Therefore, the following conclusive theorem arises from the combination of Proposition \ref{theorem:rETHIS}, Lemma \ref{lemma:easypart}, and Lemma \ref{lemma:hardpart}.

	
{\finalresult*}

%% file: CVD+A.tex
\section{FPT algorithm parameterized by ${\cvd} + |A|$}
\label{sec:cvd-with-A}
In this section, we design an \fpt algorithm for the {\ALP} problem parameterized by combining the two following parameters: the size of a cluster vertex deletion set and $|A|$. 
Our algorithm operates under the assumption that we are provided with a minimum size set $M$ as input, satisfying the conditions $A \subseteq M$ and $G-M$ forms a cluster graph. 
This assumption is justified by the fact that one can efficiently find the smallest size cluster vertex deletion set of $G-A$ in time $1.9102^k\cdot n^{\OO(1)}$ \cite{boral2016fast}.

\subparagraph*{Overview of the Algorithm:} Our algorithm starts by making an educated guess regarding the precisely ordered intersection of each path within an optimal solution ($\Pe$) with the modulator set $M$. 
This involves exploring a limited number of possibilities, specifically on the order of $f(|M|)\cdot n^{\OO(1)}$ choices.
Once we fix a choice, the problem reduces to finding subpaths (of any given path in $\Pe$) between modulator vertices satisfying certain length constraints.
To provide a formal description, let $P\in \Pe$ be a path of the form $m_1P_{1,2}m_2m_3m_4P_{4,5}m_5$ (our guess for $P\in \Pe$) where each $m_i\in M$ and each $P_{i,i+1}$ is contained in $ G-M$.
Subsequently, our algorithm proceeds to search for the subpaths $P_{1,2}$ and $P_{4,5}$ each contained in cliques of $G-M$ with endpoints adjacent to vertices $m_1,m_2$ and $m_4,m_5$ respectively.
A collection of constraints within our final \emph{integer linear program (ILP)} guarantees that the combined length of the paths $P_{1,2}$ and $P_{4,5}$ precisely matches $\ell-5$, satisfying the prescribed length requirement.
Before presenting the ILP formulation, we make informed decisions about the cliques that are well-suited and most appropriate for providing these subpaths.
Towards that, we partition the at most $|M|$ many subpaths (originating from all the paths in $\Pe$), with each partition containing subpaths exclusively from a single clique of $G-M$. Moreover, no two subpaths in separate partitions come from the same clique.
Once such a choice is fixed, we apply a color coding scheme on the cliques in $G-M$ where we color the cliques with the same number of colors as the number of sets in the mentioned partition of subpaths into sets.
With a {\em high} probability, each clique involved in the formation of $\Pe$ is assigned a distinct color. These assigned colors play a crucial role in determining the roles of the cliques in providing subpaths and, consequently, in constructing the final solution $\Pe$.
We show that among all the cliques colored with a single color, a largest size \emph{feasible} clique is an optimal choice for providing the necessary subpaths of $\Pe$.
A \emph{feasible} clique is a clique that is able to provide the necessary subpaths determined by its assigned color, barring the length requirements, and, is identified by its adjacency relation with $M$. Therefore, we keep precisely one feasible clique of the maximum size for each color and eliminate the others.
Thus the number of cliques in the reduced instance is bounded by a function $f(m)$.
Following these steps, our problem reduces to finding required subpaths (with length constraints) for which we design a set of ILP equations where the number of variables is a function of $m$. Below, we give a detailed description of our algorithm.\\

\vspace{3pt} 
\hrule 
\vspace{3pt} 
\noindent\textbf{Algorithm:}
\vspace{3pt} 
\hrule 
\vspace{2pt} 
\vspace{-10pt}
\subparagraph{\underline{Phase 1: The Guessing Phase}}
\begin{enumerate}
\item Find a cluster vertex deletion set $S$ of the minimum size in $G-A$. Then, set $M= A \cup S$.

\item Generate all $M'\subseteq M$. For a fixed $M'$, generate all its {\em ordered partitions} such that only the first and last vertices of every set of the partitions are from $A$. Let $\M=\{M'_1,\ldots,M'_{|\M|}\}$ be a fixed such partition of $M'$.

\item Without loss of generality, let $M'_i=(m_{g(i)},m_{g(i)+1},\ldots,m_{l(i)})$. For any two consecutive vertices $m_j$ and $m_{j+1}$ where $j\in[g(i),l(i)-1]$, we introduce a variable $P_{j,j+1}$. These variables serve as placeholders representing subpaths within the optimal solution that we are aiming to find. We use $\mathbb{P}$ to denote the collection of $P_{j, j+1}$.

\item We enumerate all partitions of $\mathbb{P}$. Let $\X_{\mathbb{P}}=\{\Pe_1,\Pe_2,\ldots,\Pe_x\}$ be a partition of $\mathbb{P}$ with $x=|\X_{\mathbb{P}}|$.

\item Additionally, we generate all \emph{valid} functions $h: \mathbb{P}\rightarrow \{0,1,>1\}$, i.e., we guess whether the length of each subpath is exactly 0, 1, or more than 1. The function $h$ is \emph{valid} if $h(P_{j,j+1})=0$, implies $m_j$ and $m_{j+1}$ are adjacent. The validity of $h$ concerning the other two values (1 and $>1$) is inherently assured by the presence of a feasible clique.


\item We create a function $T:\X_{\mathbb{P}}\rightarrow2^{2^{M}}$ as follows. For each part $\Pe_i \in \X_{\mathbb{P}}$, create the set 
$T(\Pe_i)=\{\{m_j,m_{j+1}\}: P_{j, j+1} \in \mathcal{P}_i \text{ and }h(P_{j,j+1})=1\}\cup \{\{m_j\}\cup\{m_{j+1}\}: P_{j, j+1} \in \mathcal{P}_i \text{ and } h(P_{j,j+1})> 1\}$.




\end{enumerate}
At the conclusion of Step 6, we have generated all the tuples denoted as $\tau=(M',\M,\mathbb{P},\X_{\mathbb{P}},h,T)$. For each specific $\tau$, we proceed to Phase 2 in order to bound the number of cliques and subsequently generate a set of ILPs.

\subparagraph{\underline{Phase 2: Bounding the number of Cliques Phase}}
\begin{enumerate}
\item  We color all the cliques in $G-M$ with $x$ many colors uniformly at random. From a set of cliques colored with color $i$, we choose a largest \emph{feasible} clique $Q_i$. A clique with color $i$ is \emph{feasible} if and only if it has $|T(\Pe_i)|$ distinct vertices, each being a neighbor to a different set in $T(\Pe_i)$. Also, we denote the above coloring function by $\mathcal{C}_{\tau}$.

\item Following the Algorithm, we construct the following set of ILPs.

\vspace{-5pt}
\begin{mdframed}[backgroundcolor=green!10,topline=false,bottomline=false,leftline=false,rightline=false]
 \begin{center}
 \vspace{-12pt}
   \begin{equation*}
\begin{array}{ll@{}rl}
{}{ILP(\tau,\mathcal{C}_{\tau}):}{}{}
  & \displaystyle\sum\limits_{j=g(i)}^{l(i)-1} x_{j,j+1} = \ell-|M'_i|, &~\forall M'_i\in \M \\
 & \displaystyle\sum\limits_{P_{j,j+1}\in \Pe_i}^{} x_{j,j+1} \leq |Q_i|, &\forall \Pe_i \in \X_\mathbb{P}\\
                 &                                                x_{j,j+1}=1, &\text{iff }~ h(P_{j,j+1})=1
\end{array}
\end{equation*}
 \end{center}
 \end{mdframed}

\end{enumerate}

\hrule 
\vspace{3pt} 


%

\vspace{-10pt}
\subparagraph*{Correctness of The Guessing Phase (Steps 1 to 6):} Let $\Pe=\{P_1,\ldots, P_p\}$ be an optimal solution of size $p$ where any $\ell$-path in $\Pe$ by definition has both its endpoints in $A$.
In the above ILP, note that the variable $x_{j, j+1}$ represents the path $P_{j, j+1}$.
The $M'$ generated in the Step 2 is $V(\Pe)\cap M$. Each $M'_i$ is the {\em ordered intersection} of a path $P_i$ with $M'$ , i.e., the sequence of vertices of $V(P_i)\cap M$ appearing in the path is given by $M_i'$ (Step 2).
In Step 3, we create the variables (corresponding to the subpaths of $\Pe$) for each pair of consecutive vertices from $M'_i$ for every $M'_i\in\M$ (Step 3). Any such subpath with a non-zero length is contained in exactly one of the cliques in $G-M$. The subpaths of $\Pe$ that come from single cliques together are denoted by the partition $\X_{\mathbb{P}}$, i.e., the subpaths (in $\Pe$) in a part of the partition are exactly the subpaths that are contained in a single clique of $G-M$ (Step 4).
We further divide these subpaths into three groups ($h^{-1}(0)$, $h^{-1}(1)$, $h^{-1}(>1)$) based on whether their lengths are exactly 0, 1, or more than 1 (Step 5).
If $h(P_{j,j+1})=0$, then the corresponding subpath has length {\em zero} implying $m_j$ and $m_{j+1}$ are adjacent in $\Pe$.  If $h(P_{j,j+1})=1$, then the corresponding subpath has length exactly one and the lone vertex in the subpath is adjacent to both $m_j$ and $m_{j+1}$ in $\Pe$. 
When $h(P_{j,j+1})>1$, the correcponding subpath has length {\em more than one} and has two vertices, one is adjacent to $m_j$ while another is adjacent to $m_{j+1}$. The set $T(\Pe_i)$ basically stores the adjacency relations (required) between the endpoints of non-zero length subpaths of $\Pe$ and $M$. The correctness of the first $6$ steps follows directly because of the fact that we exhaust all possible choices at each step. 

\vspace{-10pt}
\subparagraph*{Correctness of Phase 2:} 
We apply the {\em color-coding} scheme in the first step of second phase of  our algorithm. Each clique $Q_i \in \{Q_1,\ldots,Q_x\}$ that contains vertices from $\Pe$ gets a different color with \emph{high} probability. Moreover, each $Q_i$ colored with color $i$ exactly contains the subpaths denoted by $\Pe_i$. We compute this exact probability later in the runtime analysis of our algorithm. Notice the role of the cliques in $G-M$ is to provide subpaths of certain lengths between the vertices from $M$. And, a {\em feasible} clique of color $i$ is able to provide all the subpaths in $\Pe_i$ between the vertices of $M$, barring the length requirements.
Thus, given a feasible clique $Q'_i$ of maximum size, we can reconstruct an equivalent optimal solution $\Pe'$ in which all its paths within $\Pe_i$ are entirely contained within $Q'_i$, all the while sticking to the specified length requirements. This reconstruction can be systematically applied to guarantee the existence of an optimal solution where all its subpaths are derived from a collection of feasible cliques with the largest size available from each color class. Let $\tau$ be a correctly guessed tuple, $\mathcal{C}_{\tau}$ be a correct coloring scheme (coloring each of the $x$ cliques involved in the solution distinctly), and $Q'_i$ be a feasible clique of the largest possible size colored with color $i$ respecting the guessed tuple for each color $i$. Then, there exists an optimal solution that is entirely contained in  the subgraph $G[\bigcup_{i=1}^{x}Q'_i \cup M]$.
To obtain the desired solution for a {\yes} instance, we narrow our attention to the subgraph and formulate the specified ILP denoted as ILP($\tau,\mathcal{C}_{\tau}$). The primary objective of the ILP equations is to guarantee that every path we are seeking has an exact length of $\ell$.
The first set of constraints  enforces the specified length requirements for the subpaths, ensuring that each subpath adheres to its designated length. 
And the second set of constraints ensures that the combined total of all vertices to be utilized from a clique (across all subpaths) in a solution to the ILP does not surpass the total number of vertices within the largest feasible clique.

\vspace{-10pt}
\subparagraph*{Runtime Analysis:}
 The total number of ordered partitions generated in Step 2 is  $\OO(2^m\cdot m^m)$. In Step $3$, $|M'_i|$ can be of $\OO(m)$. Hence for a fixed $\M$, the number of permutations enumerated is of $\OO(m!)\cdot m$. Notice the number of $P_{j,j+1}~ (|\Pb|)$ is bounded by $m$. Therefore in the next step, $\X_{\Pb}$ can be partitioned in $m^m$ ways. In Step 5, each $P_{j,j+1}\in \Pb$ takes one of the three values. Thus there can be at most $3^m$ assignments for a fixed  $\Pb$. Hence total number of tuples generated at the end of Step 6 is bounded by  $\OO(2^m\cdot m^m\cdot m!\cdot m\cdot m^m \cdot 3^m)\equiv 2^{\OO(m\log m)}$.
 Since $x\leq m$, the probability that we get a coloring that  colors all the $x$ cliques properly   and distinctly is at least $\frac{1}{m!}$. Once we have a good coloring instance, we formulate the ILP($\tau,\mathcal{C}_{\tau}$) to solve the problem. Since both $|\M|$ and $|\X_{\Pb}|$ are bounded by $\OO(m)$, the ILP can be solved in time $m^{\OO(m)}$. This immediately implies a randomized \fpt algorithm running in time $2^{\OO(m\log m)}$. Notice the randomization step (Phase 2) can be  derandomized using $(m,x)$-universal family \cite{CyganFKLMPPS15}. And we have the following theorem.

{\cvdplusA*}
 

%% file: CVD+l-Aritra.tex
\section{FPT algorithm when parameterized by ${\cvd} + \ell$}
\label{sec:cvd+l}

In this section, we design an {\fpt} algorithm for {\alpp} when parameterized by the combined parameter ${\cvd} +\ell$.
Let the input graph be $G$ with two subsets $A, M \subseteq V(G)$ such that $G - M$ is a cluster graph, $|A| = a$ and $|M| = m = {\cvd}(G)$.
The considered parameter is $m + \ell$.
%
%
We denote the set of cliques in $G-M$ by $\Qq$ and the vertices in the cliques by $V_{\Qq}=\cup_{Q\in \Qq} V(Q)$.
Let $\I=(G, M, A, k, \ell)$ be a {\yes} instance of {\alpp} and let $\Pe$ be any arbitrary solution  for $\I$.
We denote the set paths in $\Pe$ that contain at least one vertex from $M$ by ${\PM}$ and the set of paths in $\Pe$ that are completely inside a clique by ${\PQ}$.
Note that $\PM\cap\PQ=\emptyset$ and ${\Pe} = {\PM} \cup {\PQ}$.

\begin{restatable}{observation}{cvdlmaxvertex}
\label{obs:cvd+l-maxvertex}
The total number of vertices present in the paths $\PM$ is at most $\ell \cdot m$, i.e. $|\cup_{P\in \PM}V(P)|\leq \ell \cdot m$.
\end{restatable}

Next, we present a marking procedure followed by a few reduction rules to bound the size of each clique.

\begin{mdframed}[backgroundcolor=cyan!5,topline=false,bottomline=false,leftline=false,rightline=false]

\begin{center}
 \vspace{-4pt}
{\bf Marking Procedure.} 
\begin{description}
\item[1.] For each vertex $u\in M$, mark $\ell m+ 1$ many of its neighbors from both $A\cap V(Q)$ and $V(Q)\setminus A$ for each clique $Q\in \Qq$. If any clique does not contain that many neighbors of $u$, we mark all the neighbors of $u$ in that clique. 
\item[2.] For each pair of vertices $u,v$ in $M$, mark $\ell m+ 1$ many common neighbors of $u$ and $v$ outside $A$, in every clique of $\Qq$. 
\item[3.] Additionally, mark $\ell m+ 1$ many vertices from both $A\cap V(Q)$ and $V(Q)\setminus A$ for each clique $Q\in \Qq$.
\end{description}
 \end{center}
 \end{mdframed}
 
In the marking procedure the upper bound on the number of marked vertices for each clique $Q$ from $A$ (in $A\cap V(Q)$) is $f_1(\ell,m)=(m+1)(\ell m+1)$ and the number of marked vertices outside $A$ (in $V(Q)\setminus A$) is $f_2(\ell,m)=(m^2+m+1)(\ell m+1)$.

\vspace{10pt}
\subparagraph{Exchange Operation:} Consider any two arbitrary paths  $P_1,P_2\in\Pe$ and $<a_1,a_2,a_3>$ and $<b_1,b_2,b_3>$ be any two subsequences of vertices in $P_1$ and $P_2$ respectively. Let $a_2$ be a neighbour of $b_1$ and $b_3$ and $b_2$ be a neighbour of $a_1$ and $a_3$. We define the operation \emph{exchange} with respect to $P_1$, $P_2$, $a_2$, and $b_2$ as follows. We create the path $P_1'$ by replacing the vertex $a_2$ with $b_2$, and we create the path $P_2'$ by replacing the vertex $b_2$ with $a_2$. Observe that $\Pe\setminus\{P_1,P_2\}\cup\{P_1',P_2'\}$ also forms a solution.  

\begin{lemma}
    There exists a solution $\Pe$ for $(G,M,A,k,\ell)$ such that all the vertices in $\PM$ are either from $M$ or are marked.\label{lemma:onlymarked}
\end{lemma}
\begin{proof}
    Suppose there is a path $P$ in $\PM$ that contains an unmarked vertex $w$. There are at most two neighbors of $w$ in $P$. We assume here that there are exactly two neighbors of $w$. The case when $w$ has only one neighbor in $P$ can be argued similarly.     
    Let the neighbors of $w$ in $P$ be $w_1$ and $w_2$. We have the following three exhaustive cases.

    \begin{description}
	\item[$w_1,w_2\in V(Q)$]: Recall that we have marked an additional $\ell m+ 1$ many vertices from  outside $A$ in each clique (vertices that $w$ may be replaced with) and from Observation~\ref{obs:cvd+l-maxvertex}, we know at most $m\ell$ many of them are contained in $\PM$.
	Thus, there is at least one marked vertex, say, $w'$  in $V(Q)$,  that is not contained in any path of $\PM$. If $w'$ is also not contained in any path of $\PQ$, we simply replace $w$ by $w'$ in $P$. If it is in a path $P'\in \PQ$, we do an \emph{exchange} operation with respect to $P$, $P'$, $w$ and $w'$ and reconstruct a new solution.

    \item[$w_1\in V(Q)$ and $w_2\in M$]: Recall that we have marked $\ell m+ 1$ many vertices   from $N(w_2)\cap V(Q)\setminus A$. From Observation~\ref{obs:cvd+l-maxvertex}, at most $m\ell$ many of them are contained in $\PM$. Hence, there is at least one marked vertex in $V(Q)\setminus A$ that is not contained in any path from $\PM$. Similar to the arguments outlined in the previous case, we replace the vertex $w$ by $w'$ when $w'$ is not contained in any path of $\PQ$, or perform an \emph{exchange} operation with respect to $P$, $P'$, $w$ and $w'$ and reconstruct a new solution when $w'$ is contained in some $P' \in \PQ$. 
	
    \item[$w_1,w_2\in M$]: Using arguments similar to the previous case, we can again replace an unmarked vertex in $\PM$ with a marked vertex and reconstruct a new solution for this case as well.
\end{description}
    
After exhaustively replacing unmarked vertices of $\PM$ (that are not in $M$), we derive a solution $\Pe$ in which paths from $\PM$ do not include unmarked vertices from the cliques.
\end{proof}

Henceforth, we seek for a solution $\Pe$ for $(G,M,A,k,\ell)$ such that all the vertices in any path of $\PM$ are either from $M$ or marked. 
Next, we have the following reduction rule.

\begin{restatable}{redrule}{rulecvdlpath}
\label{rule:cvd+l-path}
If there exists a clique $Q$ containing a pair of unmarked vertices $u,v\in A$ and a set $X$ of $(\ell -2)$ unmarked vertices outside $A$, then delete $u,v$ along  with $X$ and return the reduced instance $(G-\{X\cup \{u,v\}\},M,A\setminus \{u,v\},k-1,\ell)$.
\end{restatable}



We prove the safeness of the reduction rule below.

\begin{restatable}{lemma}{cvdlpath}
\label{lemma:cvd+l-path}
$(G,M,A,k,\ell)$ is a \yes instance if and only if $(G-\{X\cup \{u,v\}\},M,A\setminus \{u,v\},k-1,\ell)$ is a \yes instance.
\end{restatable}

\begin{proof}
    If $(G-\{X\cup \{u,v\}\},M,A\setminus \{u,v\},k-1,\ell)$ is a \yes instance with a solution $\Pe^R$, then $(G,M,A,k,\ell)$ is also a \yes instance as $\Pe^R$ along with the path formed by $X\cup \{u,v\}$ forms a solution to the instance.

    Conversely, let $(G,M,A,k,\ell)$ be a \yes instance with a solution $\Pe$. Now, we will obtain a solution $\Pe'$ for $G-\{X\cup \{u,v\}\}$ of size at least $k-1$. We denote the paths in $\Pe$ that intersect with $X\cup \{u,v\}$ by $\Pe_X$ (with slight abuse of notation). 
    If $|\Pe_X|\leq 1$ then $\Pe'=\Pe\setminus\Pe_X$ is the desired solution. From now on, we assume that $|\Pe_X|>1$.  
    Observe that as the vertices in $X\cup\{u,v\}$ are unmarked, $\Pe_X\cap\PM=\emptyset$ (from Lemma~\ref{lemma:onlymarked}). 
    We can reconstruct a new solution $\Pe_1$ from $\Pe$ where exactly one path $P$ in $\Pe_1$ intersects $X\cup\{u,v\}$, and the rest of the paths in $\Pe_1$ do not intersect with $X\cup\{u,v\}$, by repeatedly utilizing the \emph{exchange} operation among the paths in $\Pe_X$. Observe that $\Pe'=\Pe_1\setminus\{P\}$ is the desired solution.
    Thus, the claim holds. 
\end{proof}

Note that the upperbound on the number of marked vertices from $A\cap V(Q)$ is $f_1(\ell,m)=(m+1)(\ell m+1)$ and the number of marked vertices $V(Q)\setminus A$ is $f_2(\ell,m)=(m^2+m+1)(\ell m+1)$. And, after exhaustive application of Reduction Rule \ref{rule:cvd+l-path}, in any clique $Q$, either there are at most $\ell-3$ unmarked vertices in $V(Q)\setminus A$ or at most one unmarked vertex in $A\cap V(Q)$.

\begin{description}
    
    \item[Case{\it (i)}:] There is at most one unmarked vertex in $A\cap V(Q)$. 
    
    \item[Case{\it (ii)}:] There are at most $\ell-3$ unmarked vertices in $V(Q)\setminus A$.
\end{description}

Based on the aforementioned cases, we introduce two reduction rules—one for each case—that help us limit the overall number of unmarked vertices in $Q$, thereby bounding the size of each clique in $G-M$. First we consider the Case (\emph{i}) when the number of unmarked vertices form $A\cap V(Q)$ is bounded by one and bound the number of the unmarked vertices in $V(Q)\setminus A$ with the following reduction rule.


\begin{restatable}{redrule}{rulecvdlnonA}
\label{rule:cvd+l-non-A}
If there exists a clique $Q$ containing at most one unmarked vertex from $A$ and at least $(f_1(\ell,m)+1)\cdot \frac{\ell}{2}+1$ unmarked vertices outside $A$, then delete one unmarked vertex $u\in V(Q)\setminus A$ and return the reduced instance $(G-\{u\},M,A,k,\ell)$.
\end{restatable}

Let $G'=G-\{u\}$ be the new graph following an application of Reduction Rule~\ref{rule:cvd+l-non-A}. We prove the safeness of the reduction rule in the following lemma.

\begin{restatable}{lemma}{cvdlnonA}    
\label{lemma:cvd+l-non-A}
 $(G,M,A,k,\ell)$ is a \yes instance if and only if $(G',M,A,k,\ell)$ is a \yes instance.
\end{restatable}

\begin{proof}
    If $(G',M,A,k,\ell)$ is a {\yes} instance, then $(G,M,A,k,\ell)$ is a \yes instance since $G'$ is a subgraph of $G$.
 Conversely, suppose $(G,M,A,k,\ell)$ is a \yes instance, and ${\Pe}$ is a solution.
 If $u$ does not belong to any path in $\Pe$, then $\Pe$ is a solution to $(G',M,A,k,\ell)$ as well. Otherwise, let $P\in\PQ$ contain $u$. This is true since any unmarked vertex can only be used in a path in $\PQ$. But any such path uses exactly 2 vertices from $V(Q)\cap A$. Hence we can upper bound the number of unmarked vertices outside $A$ that are contained in $\PQ$ and hence $\Pe$ by $(f_1(\ell,m)+1)\cdot \frac{\ell}{2}$. Hence, there is at least one unmarked vertex $u'\neq u$ in $V(Q)\setminus A$ which is not used by any path in $\Pe$. We replace $u$ with $u'$ in $P$ to get a desired solution to $(G',M,A,k,\ell)$.
\end{proof}

For the Case (\emph{ii}) when the number of unmarked vertices form $ V(Q)\setminus A$ is bounded by $\ell - 3$ and we bound the number of the unmarked vertices in $V(Q)\cap A$ with the following reduction rule.

\begin{restatable}{redrule}{cvdlA}
\label{rule:cvd+l-A}
If there exists a clique $Q$ containing at most $\ell -3$ unmarked vertices from $V(Q)\setminus A$ and at least $(f_2(\ell, m)+(\ell-3))\cdot \frac{1}{\ell-2}+1$ many unmarked vertices in $A$, then delete an unmarked vertex $u\in A\cap Q$ and return the reduced instance $(G-\{u\},M,A\setminus \{u\},k,\ell)$.
\end{restatable}

\begin{proof}
If $(G-\{u\},M,A\setminus \{u\},k,\ell)$ is a {\yes} instance, then $(G,M,A,k,\ell)$ is trivially a \yes instance since $G'$ is a subgraph of $G$.
 Conversely, suppose $(G,M,A,k,\ell)$ is a \yes instance, and ${\Pe}$ is a solution.
 If $u$ does not belong to any path in $\Pe$, then $\Pe$ is a solution to $(G',M,A,k,\ell)$ as well. Otherwise, let $P\in\PQ$ contain $u$. This is true since any unmarked vertex can only be used in a path in $\PQ$. But any such path uses exactly $\ell-2$ vertices from $V(Q)\cap A$. Hence we can upper bound the number of unmarked vertices from $A$ that are contained in $\PQ$ and hence $\Pe$ by $(f_2(\ell, m)+(\ell-3))\cdot \frac{1}{\ell-2}$. Hence, there is at least one unmarked vertex $u'\neq u$ in $V(Q)\cap A$ which is not used by any path in $\Pe$. We replace $u$ with $u'$ in $P$ to get a desired solution to $(G',M,A,k,\ell)$. 
\end{proof}

After exhaustive application of Reduction Rules~\ref{rule:cvd+l-non-A} and~\ref{rule:cvd+l-A}, the upper bound on the number of vertices of different types in each clique is as follows:
\begin{itemize}
    \item Marked vertices in $A$: $f_1(\ell,m)=(m+1)(\ell m+1)$
    \item Marked vertices in $V(Q)\setminus A$:  $f_2(\ell,m)=(m^2+m+1)(\ell m+1)$
    \item Unmarked vertices in $A$: $(f_2(\ell, m)+(\ell-3))\cdot \frac{1}{\ell-2}+2$
    \item Unmarked vertices in $V(Q)\setminus A$: $(f_1(\ell,m)\cdot \ell+1$
\end{itemize}
Hence the total number of vertices in each clique is bounded by  $\OO(\ell^2 m^2+\ell m^3)$.

\subparagraph{Equivalent cliques:} Now we aim to bound the number of cliques by introducing the concept of \emph{equivalent} cliques. Two cliques $Q_i$ and $Q_j$, are equivalent (belong to the same {\em equivalent class}) if and only if the number of vertices from the cliques that are in $A$, and that are outside $A$ with an exact neighborhood of $M'\subseteq M$ is same for each $M' \in 2^{M}$. Two cliques $Q_i, Q_j$ in an equivalence class are essentially {\em indistinguishable} from each other, i.e.,  there is a bijective mapping $g_{ij}: V(Q_i)\mapsto V(Q_j)$, so that $N(u)\cap M=N(g(u)) \cap M$, for all $u\in V(Q_i)$. This fact is crucial in the construction of our next reduction rule.
 Observe that the number of equivalence classes is at most $\OO(\ell^2 m^2+\ell m^3)^{2^m}$=$f(\ell, m)$.
The following reduction rule bounds the number of cliques in each equivalent class.

\begin{restatable}{redrule}{cvdlcliquebound}
\label{rule:cvd+l-clique-bound}
If there exists an equivalent class $\Ce$ with at least $\ell m+1$ cliques, then delete one of the cliques $Q_i\in \Ce$ and return the reduced instance $(G-Q_i,M,A\setminus (A\cap V(Q_i)),k-x_i,\ell)$ where, $x_i= \min\Bigl\{\frac{|A\cap V(Q_i)|}{2}, \frac{|V(Q_i)\setminus A|}{\ell - 2}\Bigr\}$. 
\end{restatable}

\begin{lemma}
\label{lem:cvd-plus-l-eqv}
   $(G,M,A,k,\ell)$ is a \yes instance if and only if $(G-Q_i,M,A\setminus (A\cap V(Q_i)),k-x_i,\ell)$ is a \yes instance.
\end{lemma}

\begin{proof}
    In the forward direction, let $(G,M,A,k,\ell)$ be a~\yes instance.
    Recall that the number of the vertices contained in paths of $\Pe_M$ for any optimal solution $\Pe$ is bounded by $\ell s$. Thus, there are at most $\ell m$ many cliques in total and also from any equivalence class that has vertices in paths from $\Pe_M$. Let $Q_j$ be one such clique in the equivalence class $\Ce$ that does not contain any vertex in the paths from $\Pe_M$. From the definition of an equivalence class, it is evident that the two cliques $Q_i, Q_j$ in the equivalence class $\Ce$ are {\em indistinguishable} from each other, i.e.,  there is a bijective mapping $g_{ij}: V(Q_i)\mapsto V(Q_j)$, so that $N(u)\cap M=N(g(u)) \cap M$, for all $u\in V(Q_i)$. Let $X_i=V(\Pe)\cap V(Q_i)$ and $X_j=V(\Pe)\cap V(Q_j)$, i.e, the set of vertices from $Q_i$ and $Q_j$ that are used in paths from $\Pe$, respectively. 
    We construct an alternate solution, $\Pe'$, where we replace $X_i$ with $g_{ij}(X_i)$ and $X_j$ with $g_{ij}^{-1}(X_j)$ in $\Pe$. Since $X_j\cap M=\emptyset$,  we have $g_{ij}^{-1}(X_j)\cap M= \emptyset$. Therefore in $\Pe'$, there is no path that contain vertices from both $M$ and $V(Q_i)$. In other words, vertices in $Q_i$ can only be contained in paths from $\Pe'\setminus \Pe'_M$ (paths that are completely contained inside the clique). And, the number of such paths is bounded by $x_i=\min\Bigl\{\frac{|A\cap V(Q_i)|}{2}, \frac{|V(Q_i)\setminus A|}{\ell - 2}\Bigr\}$. Hence $(G-Q_i,M,A\setminus (A\cap V(Q_i)),k-x_i,\ell)$ is a \yes instance. 
    
    In the reverse direction, let $(G-Q_i,M,A\setminus (A\cap V(Q_i)),k-x_i,\ell)$ be a \yes instance with a solution $\mathcal{P}$. But there are $x_i$ many paths (say $\mathcal{P}_i$) that are completely contained in $Q_i$.Hence, $\mathcal{P}\cup\mathcal{P}_i$ is a set of $k$ vertex-disjoint $(A,\ell)$-paths contained in $G$, making $(G,M,A,k,\ell)$ a \yes instance.
\end{proof}

\noindent
After exhaustively applying all the aforementioned reduction rules, the following bounds hold.
\begin{itemize}
    \item The number of vertices in each clique is bounded by  $\OO(\ell^2 m^2+\ell m^3)$.
    \item The number of equivalence classes is at most $\OO(\ell^2 m^2+\ell m^3)^{2^m}$.
    \item The number of cliques in each equivalence class is at most $\ell m+1$.
\end{itemize}
    
Consequently, the size of the reduced instance is upper-bounded by a computable function of $\ell$ and $m$, thus directly implying the following theorem.

{\cvdlShortALPPresult*}

%% file: vertex-cover.tex
\section{Polynomial kernel when parameterized by the vertex cover number}
\label{sec:ALPP-vc}

In this section, we design a quadratic vertex kernel for {\ALP} parameterized by the size of a minimum vertex cover of the graph.
The input instance contains an undirected graph $G = (V, E)$, two sets $A, M \subseteq V(G)$ such that $|M| \leq 2{\vc}(G)$, and integers $k$ and $\ell$.
For the ease of notation, we assume $|M| = s$.
%
%
The kernelization algorithm is as follows.
Let $I = V(G - M)$, $A_I = I \cap A$ and $A_M = M \cap A$.
Observe that the vertices of $I \setminus A$ can only be the internal vertices of any $A$-path.
Hence, any $A$-path of length $\ell$ can be of the following three different types.
The {\em type-1} is when it starts and ends at two distinct vertices of $A_M$.
The {\em type-2} is when it starts in $A_M$ but ends in $A_I$.
The {\em type-3} is when it starts and ends in $A_I$.

Observe that any pendant vertex $u \in I \setminus A$, cannot be present in any $A$-path. 
Hence we can safely delete all such vertices.
Note that every $A$-path of $G$ can have exactly two vertices of $A$, in addition to having at most $s$ vertices from $M \setminus A$ and at most $s-1$ additional vertices from $I \setminus A$.
Therefore, any $A$-path of $G$ has at most $2s + 1$ vertices.
Since our problem involves packing $A$-paths with exactly $\ell$ vertices each, we have that $\ell \leq 2s + 1$.
This gives us the following observation.

\begin{observation}
\label{obs:path-length-bound}
Any $A$-path has at most $2s + 1$ vertices. 
\end{observation}

Our kernelization algorithm has two phases.
The first phase is to design reduction rules that provide an upper bound on the number of vertices in $A_I$ and the second phase is to design reduction rules that provide an upper bound on the number of vertices in $I \setminus A$.
If $\ell \leq 4$, then {\ALP} can be solved in polynomial-time (Theorem 3.3 of \cite{BelmonteHKKKKLO22}).
Therefore, we can assume that $\ell \geq 5$.

\subsection{Bounding the size of $A_I$}

We start by proving the following statement for packing $A$-paths of length exactly $\ell$ for an instance $(G, M, A, k, \ell)$.
From now onward, $\ell \geq 5$.

Let $H_1$ be an auxiliary bipartite graph with bipartition $(L_1 \uplus R_1)$ such that $L_1 = M \setminus A$, $R_1 = A_I$, and $E(H_1)=\{uv \mid uv \in E(G), u \in L_1\text{, and } v \in R_1\}$. We apply the following reduction rule when $|A_I| > 2|M \setminus A|$.

\begin{redrule}
\label{rule:compressing-A_I-part}
If $|A_I| > 2|M \setminus A|$, then we construct the auxiliary bipartite graph $H_1$ as described above.
We apply Proposition \ref{lemma:new-expansion-lemma} (New $q$-Expansion Lemma) with $q = 2$ and get $M \subseteq E(H_1)$, $\widehat{L_1} \subseteq L_1, \widehat{R_1} \subseteq R_1$ and a $2$-expansion $\MM$ from $\widehat{L_1}$ to $\widehat{R_1}$.
Consider $V(\MM)$ the vertices of $\widehat{L_1}$ and $\widehat{R_1}$ that are saturated by $\MM$.
Pick one arbitrary vertex $v \in \widehat{R_1}$ that is not in $V(\MM)$.
Delete $v$ from $G$ and return the reduced instance is $(G' = G - v, M, A' = A \setminus \{v\}, \ell, k)$.
\end{redrule}



Before proving the safeness of the above reduction rule, we introduce some additional terminologies.
An $(A, \ell)$-path $P$ {\em respects} $\MM$ if for every edge $uv$ of $P$, $uv \in \MM$ if and only if $u \in \widehat{L_1}$ and $v \in \widehat{R_1}$.
An $(A, \ell)$-path $P$ {\em violates} $\MM$ it does not respect $\MM$.
Equivalently, if an $(A, \ell)$-path $P$ {\em violates} $\MM$, then there exists an edge $uv$ of $P$ such that $u \in \widehat{L_1}, v \in \widehat{R_1}$ but $uv \notin \MM$.
If $P$ is a type-2 or type-3 $(A, \ell)$-path and $P$ contains $v \in \widehat{R_1}$ and $u \in \widehat{L_1}$ such that $uv \notin \MM$, then we say that $u$ and $v$ are {\em $\MM$-breaching} vertices of $\widehat{L_1}$ and $\widehat{R_1}$, respectively.
If an $(A, \ell)$-path $P$ contains an edge $uv$ such that $v \in \widehat{R_1}, u \in \widehat{L_1}$ and $uv \in \MM$, then we say that $v$ is {\em $\MM$-obedient} vertex of $P$.
%

Consider an arbitrary feasible solution $\ccP$ of $G$.
If for every $P \in \ccP$ of type-2 or type-3, it holds that $P$ respects $\MM$, then $\ccP$ is said to be a solution that {\em $\MM$-respecting}.
The following two observations on type-3 and type-2 $(A,\ell)$-paths hold true.

\begin{lemma}
\label{lemma:M-breaching-L1-R1-usage}
Suppose that $|A_I| > 2|M \setminus A|$, and we apply Proposition \ref{lemma:new-expansion-lemma} with $q=2$ on $H_1$ and obtain $\widehat{L_1} \subseteq L_1$ and $ \widehat{R_1} \subseteq R_1$ such that
\begin{itemize}
	\item there is a 2-expansion $\MM$ of $\widehat{L_1}$ to $\widehat{R_1}$,
	\item  $\widehat{R_1}$ has no neighbor outside $\widehat{L_1}$ in $H_1$, and
	\item $|R_1 \setminus \widehat{R_1}| \leq 2|L_1 \setminus \widehat{L_1}|$.
\end{itemize} 
Furthermore, assume that $\ccP$ be a collection of $k$ vertex-disjoint $(A,\ell)$-paths of $G$ each violating $\MM$.
Then, for every $u \in \widehat{R_1}$, if $u$ is $\MM$-breaching vertex such that $u \in P$ for some $P \in \ccP$ and $v$ is the vertex in $P$ that appears after $u$, then $v$ is an $\MM$-breaching vertex of $\widehat{L_1}$.
Moreover, if $vx \in \MM$, then $x$ is $\MM$-breaching or $x$ is not present in any path of $\ccP$.
\end{lemma}

\begin{proof}
Suppose that $|A_I| > 2|M \setminus A|$, and we apply Proposition \ref{lemma:new-expansion-lemma} with $q=2$ on $H_1$ and obtain $\widehat{L_1} \subseteq L_1, \widehat{R_1} \subseteq R_1$ and a 2-expansion $\MM$ of $\widehat{L_1}$ to $\widehat{R_1}$.
Furthermore, assume that $\ccP$ be a collection of pairwise-disjoint $(A,\ell)$-paths of $G$ each violating $\MM$.

For the first item, consider a vertex $u \in \widehat{R_1}$ that is $\MM$-breaching and $u \in P$ for some $P \in \ccP$.
Furthermore, we assume that the vertex appearing after $u$ in $P$ is $v$.
As $u \in \widehat{R_1}$ is $\MM$-breaching and $v$ is a neighbor of $u$, then $v \in \widehat{L_1}$ (due to the reason that $N_H(\widehat{R_1}) \subseteq \widehat{L_1}$.
As $u$ is $\MM$-breaching, it must be that $uv \notin \MM$.
Therefore, $v$ is also $\MM$-breaching vertex from $\widehat{L_1}$.

For the second item, note that there is a 2-expansion from $\widehat{L_1}$ to $\widehat{R_1}$ and let $vx \in \MM$.
Then, observe that $v \in P$ but $vx \notin P$.
Hence, either $x$ must be in some $P'$ such that $P' \in \ccP$ or $x$ is not in any path of $\ccP$.
In the latter case, we are done.
In the former case, the vertex appearing next to $x$ in $P'$ is some $w \in \widehat{L_1}$ as $N_H(x) \subseteq \widehat{L_1}$.
Then, $xw \notin \MM$, implying that $x$ is $\MM$-breaching, completing the proof of the lemma.
\end{proof}

Suppose in a feasible solution $\ccP$, there are $\widehat k$ vertex-disjoint $(A, \ell)$-paths that are of type-2 or type-3 with one endpoint in $\widehat R_1$.
Then, the following lemma ensures that $\ccP$ can be converted into a feasible solution $\widehat \ccP$ such that $|\widehat \ccP| = |\ccP|$ and $\widehat{\ccP}$ is an $\MM$-respecting solution.

\begin{figure}[t]
\centering
	\includegraphics[scale=0.24]{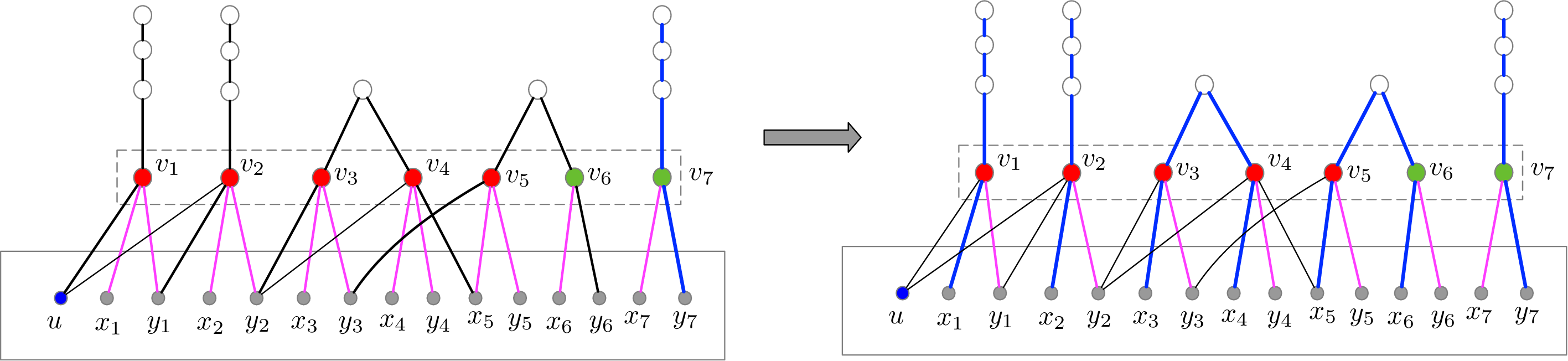}
	\caption{An example for $\ell = 5$ and $\widehat L_1 = \{v_j: j \leq 7\}$ and for every $j \in [6]$, $v_j x_j, v_j y_j \in \MM$. In the left part of the figure, paths highlighted with blue color respect $\MM$ but the paths with black color violate $\MM$. Note that $u, y_1, y_2, y_3, y_5$ are $\MM$-breaching vertices and $y_6$ and $y_7$ are $\MM$-respecting vertices. All paths that violate $\MM$ in the left figure are converted into paths that respect $\MM$ in the right part are converted to paths that respect $\MM$.}
\label{fig:expansion-part-1}
\end{figure}

\begin{lemma}
\label{lemma:vc-1-expansion-respect}
There exists a feasible solution $\ccP^*$ to $(G, M, A, \ell, k)$ that is $\MM$-respecting.
\end{lemma}

\begin{proof}
Suppose that $\ccP$ is a collection of $k$ vertex-disjoint $(A, \ell)$-paths for a {\yes} instance $(G, M, A, \ell, k)$.
Let $P_1,\ldots,P_r \in \ccP$ be the paths that are of type-2 or type-3 but violate $\MM$.
We replace $P_1,\ldots,P_r$ with $P_1^*,\ldots,P_r^*$ respectively in $\ccP$ as follows such that for every $i \in [r]$, $P_i^*$ respects $\MM$.
We keep the other paths of $\ccP$ unchanged.
For the ease of notation, let $\cQ = \ccP \setminus \{P_1,\ldots,P_r\}$.

As we assume that for every $i \in [r]$, $P_i$ is $\MM$-breaching, therefore, for every $i \in [r]$, $P_i$ has a vertex $u_i \in \widehat{R_1}$ that is $\MM$-breaching.
As $\widehat{R_1} \subseteq A_I$, there can be at most 2 vertices of $P_i$ that are in $\widehat{R_1}$.
Hence, there are at most 2 vertices of $P_i$ in $\widehat{R_1}$ that are $\MM$-breaching.
In the first step, for every $i \in [r]$, we construct $P_i'$ by deleting deleting the vertices of $\widehat{R_1}$ from $P_i$ that are $\MM$-breaching.
It is not very hard to observe that $P_1',\ldots,P_i'$ are vertex-disjoint.

Consider $v_i$ the vertex appearing after $u_i$ in $P_i$ and $u_i$ is $\MM$-breaching.
Then, due to Lemma \ref{lemma:M-breaching-L1-R1-usage}, $v_i$ is an $\MM$-breaching vertex of $P_i$.
As we delete $u_i$ from $P_i$ to construct $P_i'$, the vertex $u_i$ cannot be present in any path from $P_1',\ldots,P_r'$.
Also, $u_i$ cannot be present in any path of $\cQ$.

Now, we consider $v_i \in \widehat{L_1}$ that is $\MM$-breaching.
As there is a 2-expansion of $\widehat{L_1}$ onto $\widehat{R_1}$, there are $x_i, y_i \in \widehat{R_1}$ such that $x_i v_i, y_i v_i \in \MM$.
Then, due to Lemma \ref{lemma:M-breaching-L1-R1-usage}, $x_i$ is $\MM$-breaching or $x_i$ is not part of any $P_1,\ldots,P_r$.
If $x_i$ is $\MM$-breaching, then $x_i$ is part of some $P_j$ ($j \in [r]$) and $x_i$ is deleted from $P_j$ to construct $P_j'$.
Hence, $x_i \notin P_i'$ for any $i \in [r]$.
If $x_i$ is not part of any $P_j$ ($j \in [r]$), then $x_i$ is not part of any $P_j'$ ($j \in [r]$).

Hence, our construction of $P_1',\ldots,P_r'$ ensures that $P_i'$ has a vertex $v_i \in \widehat{L_1}$ that is $\MM$-breaching in $P_i$.
Moreover, if $v_i x_i, v_i y_i \in \MM$, then $x_i, y_i \notin P_j'$ for any $j \in [r]$.
Observe that $x_i, y_i$ cannot be part of any $(A, \ell)$-path in $\cQ$, because every $(A,\ell)$-path of $\cQ$ is $\MM$-respecting.
Hence, we add $x_i$ into $P_i'$ and complete the construction of $P_i^*$.
Clearly, $P_i^*$ respects $\MM$.
We refer to Figure \ref{fig:expansion-part-1} for an illustration.
It is not hard to observe that $P_i^*$ is disjoint with any path of $\cQ$.

Note that if $v_i \in \widehat{L_1}$ is an $\MM$-breaching vertex of $P_i$ and $v_j \in \widehat{L_1}$ is an $\MM$-breaching vertex of $P_j$, then they are distinct vertices.
Due to the existence of 2-expansion of $\widehat{L_1}$ onto $\widehat{R_1}$, it must be that $v_j x_j, v_j y_j \in \MM$.
Moreover, $x_i, y_i, x_j, y_j$ are 4 distinct vertices.
Hence, $P_1^*,\ldots,P_r^*$ are pairwise vertex-disjoint.
As $x_i$'s are in $\widehat{R_1}$ and $\widehat{R_1} \subseteq A_I$, it follows that $P_i^*$ is an $(A,\ell)$-path of $G$.
This provides us a set $\cQ \cup \{P_1^*,\ldots,P_r^*\}$ of $k$ pairwise vertex-disjoint $(A, \ell)$-paths of $G$ that is $\MM$-respecting.
\end{proof}

As a remark, we would like to illustrate that in the example of Figure \ref{fig:expansion-part-1}, we have to replace all the paths with black color by paths having edges $v_j w_j^1$ for $j \leq 7$.

\begin{lemma}
\label{lemma:vc-compressing-A_I}
Let $(G', M, A', k, \ell)$ be the reduced instance after applying Reduction Rule \ref{rule:compressing-A_I-part}.
Then, $(G, M, A, k, \ell)$ is a yes-instance if and only if $(G', M, A', k, \ell)$ is a yes-instance.
\end{lemma}

\begin{proof}
The backward direction ($\Leftarrow$) is trivial, because $G'$ is a subgraph of $G$.

For the forward direction ($\Rightarrow$), assume that $(G, M, A, k, \ell)$ is a yes-instance.
Then, it follows from Lemma \ref{lemma:vc-1-expansion-respect} that there exists a set $\ccP$ of $k$ vertex-disjoint $(A,\ell)$-paths that is $\MM$-respecting.
We consider the vertex $v \in \widehat{R_1}$ that is deleted from $G$ to get $G'$.
As $v \notin V(\MM)$, for every $u \in \widehat{L_1}$ such that $uv \notin \MM$.
Since $\ccP$ is $\MM$-respecting, no path of $P$ can contain any edge $uv$ such that $u \in \widehat{L_1}$.
As $N_H(v) \subseteq \widehat{L_1}$, $v$ cannot be present in any path of $\ccP$.
Hence, $\ccP$ is a feasible solution to $(G', M, A, k,\ell)$, implyin ghat $(G', M, A, k, \ell)$ is a yes-instance.
 \end{proof}

An exhaustive application of Reduction Rule \ref{rule:compressing-A_I-part} ensures that $|A_I| \leq 2|M \setminus A| \leq 2s$.

\subsection{Bounding the size of $I \setminus A$ and putting things together}

Now, we design our next important reduction rule to achieve an upper bound on the number of vertices in $I \setminus A$.
Since the graph has no pendant vertices in $I \setminus A$, it holds true that every $u \in I \setminus A$ has at least two neighbors in $S$.
Hence, if a vertex $u \in I \setminus A$ is present in an $(A, \ell)$-path, then $u$ is an internal vertex of an $(A, \ell)$-path.
We construct an auxiliary bipartite graph $H_2 = (L_2 \uplus R_2, E)$ and subsequently provide an additional reduction rule that helps us to reduce the number of vertices in $I \setminus A$.
We use
\begin{itemize}
	\item $L_2$ to denote the collection of all unordered pairs of vertices from $M$, i.e. $L_2 = {{M}\choose{2}}$,
	\item $R_2 = I \setminus A$, and
	\item A pair $\{x, y\} \in L_2$ is adjacent to a vertex $u \in R_2$ in $H_2$ if $\{x, y\} \subseteq N_G(u)$.
\end{itemize}

This completes the construction of $H_2$.
Using the bipartite graph $H_2$, we would like to argue that for every pair $\{x, y\}$ of vertices from $M$, it is sufficient to keep a few of them into $I \setminus A$. 
Our next reduction rule ensures this.

\begin{redrule}
\label{rule:vc-compressing-I-minus-A}
If $|R_2| > 2{{|M|}\choose{2}}$, then construct $H_2$ as described above and apply Proposition \ref{lemma:new-expansion-lemma} with $q = 2$ to obtain $\widehat{L_2} \subseteq L_2, \widehat{R_2} \subseteq R_2$ and a $q$-expansion ${\cR}$ from $\widehat{L_2}$ to $\widehat{R_2}$.
If $u \in \widehat{R_2}$ but $u \notin V({\cR})$, then delete $u$ from $G$.
The new instance is $(G - \{u\}, M, A, k, \ell)$.
\end{redrule}

Before we prove the safeness of the reduction rule, we define a few additional terminologies, and prove some structural lemmata that are central to the safeness of this reduction rule.
Consider an $(A, \ell)$-path $P$ that contains $u \in R_2 (= I \setminus A)$.
Observe that  $u$ is an internal vertex of an $A$-path and must use two additional vertices $x, y \in M$ such that $\{x, y\} \in L_2$ and $xu, uy \in E(P)$.
Then, it implies that $u$ is adjacent to the pair $\{x, y\}$ in $H_2$.
In this part, for convenience, we make abuse of notation and denote an edge $vw$ of $H_2$ as $(v, w)$.
For $X \subseteq L_2$ and $Y \subseteq R_2$, we use $E(X, Y)$ to denote the edges of $H_2$ with one endpoint in $X$ and other endpoint in $Y$. 
We say that an $A$-path $P$ {\em occupies} an edge $(\{x, y\}, u)$ of $H_2$ if $xu, uy \in E(P)$.
Equivalently, we say that an edge $(\{x, y\}, u)$ of $H_2$ is {\em occupied by $P$} if $xu, uy \in E(P)$.

Now, we consider $\widehat{L_2} \subseteq L_2$ and $\widehat{R_2} \subseteq R_2$ obtained from applying Reduction Rule \ref{rule:vc-compressing-I-minus-A}, and focus on the collection of edges in $E(\widehat{L_2}, \widehat{R_2})$, and the edges in the $2$-expansion $\cR$ from $\widehat{L_2}$ to $\widehat{R_2}$.
In order to argue the safeness of the above reduction rule, we define a few additional terminologies, and prove some crucial characteristics of a feasible solution.
Let $P$ be an $A$-path. 
If every edge $(\{x, y\}, u)$ of $E(\widehat{L_2}, \widehat{R_2})$ that is occupied by $P$ is an edge in $\cR$, then $P$ is an {\em $\cR$-respecting} $A$-path.
Otherwise, $P$ is an {\em $\cR$-violating} $A$-path.
The following observation follows from the definition of $\cR$-respecting and $\cR$-violating $A$-path.

\begin{observation}
\label{obs:violate-R-property}
If an $A$-path $P$ is $\cR$-violating, then there exists an edge $(\{x, y\}, u)$ of $E(\widehat{L_2}, \widehat{R_2})$ that is occupied by $P$ but not in $\cR$.
\end{observation}

The following lemma holds true the proof of which follows from the structural properties of $H_2$ and the characteristics of $A$-paths.

\begin{lemma}
\label{lemma:vc-pair-path-property}
Let $P$ be an $A$-path in $G$.
If $P$ occupies two distinct edges $e_1, e_2 \in E(H_2)$, then $e_1$ and $e_2$ do not share any endpoint in $H_2$.
Similarly, if $P_1$ and $P_2$ are vertex-disjoint $A$-paths such that $e_1$ is occupied by $P_1$ and $e_2$ is occupied by $P_2$, then $e_1$ and $e_2$ do not share any endpoint in $H_2$.
\end{lemma}

\begin{proof}
Let $e_1, e_2 \in E(H_2)$ be two distinct edges that are occupied by an $A$-path $P$ in $G$.
Let $e_1 = (\{x_1,y_1\}, u_1)$ and $e_2 = (\{x_2,y_2\}, u_2)$.
Since $x_1 u_1, u_1 y_1, x_2 u_2, y_2 u_2$ are the edges of $P$ in $G$, it follows that that $\{x_1,y_1\} \neq \{x_2,y_2\}$ as otherwise $P$ would be a walk but not a path.
Similarly, if $u_1 = u_2$, then also $P$ is a walk but not a path, leading to a contradiction.
Hence, $e_1$ and $e_2$ do not share any endpoint in $H_2$.

When $e_1 \in E(H_2)$ is occupied by an $A$-path $P_1$ and $e_2 \in E(H_2)$ is occupied by an $A$-path $P_2$ such that $P_1$ and $P_2$ are vertex-disjoint, we consider the endpoints of $e_1$ and $e_2$.
Let $e_1 = (\{x_1, y_1\}, u_1)$ and $e_2 = (\{x_2, y_2\}, u_2)$.
Then, $x_1 u_1$ and $u_1 y_1$ are edges of $P_1$ and $x_2 u_2$, and $u_2 y_2$ are the edges of $P_2$.
Since $P_1$ and $P_2$ are vertex-disjoint, $x_1, u_1, y_1$ do not appear in $P_2$ and $x_2, u_2, y_2$ do not appear in $P_1$.
Hence, $e_1$ and $e_2$ must not share any endpoint in $H_2$.
\end{proof}

Let $\ccP$ be a collection of $k$ vertex-disjoint $(A, \ell)$-paths in $G$.
If every $(A, \ell)$-path of ${\ccP}$ is $\cR$-respecting, then we are done.
Hence, consider the paths $Q_1,\ldots, Q_r \in {\ccP}$ that are $\cR$-violating.
Then, for every $i \in [r]$, there is an edge of $E(\widehat{L_2}, \widehat{R_2})$ occupied by $Q_i$ but not appearing in $\cR$.
If $(\{x, y\}, u) \in E(\widehat{L_2}, \widehat{R_2}) \setminus \cR$ is occupied by $Q_i$ for some $i \in [r]$, then $\{x, y\}$ and $u$ are {\em badly occupied} by ${\ccP}$.
In addition, consider a pair $\{x, y\} \in \widehat{L_2}$.
Due to Proposition \ref{lemma:new-expansion-lemma} with $q = 2$, there is a 2-expansion of $\widehat{L_2}$ onto $\widehat{R_2}$.
Hence, for every $\{x, y\} \in \widehat{L_2}$, there are 2 vertices $w, z \in \widehat{R_2}$ such that $(\{x, y\}, w), (\{x, y\}, z) \in \cR$.
We say that $w$ and $z$ are {\em $\cR$-saturated} with $\{x, y\} \in \widehat{L_2}$ and $\{x, y\}$ is the {\em $\cR$-neighbor} of $w$ and $z$.
It is clear that for any vertex of $\widehat{R_2}$ saturated by $\cR$, its $\cR$-neighbor is unique and appears in $\widehat{L_2}$.
When the set $\widehat{L_2}$ is clear from the context, we use $\{x, y\}$. 

\begin{figure}[t]
\centering
	\includegraphics[scale=0.3]{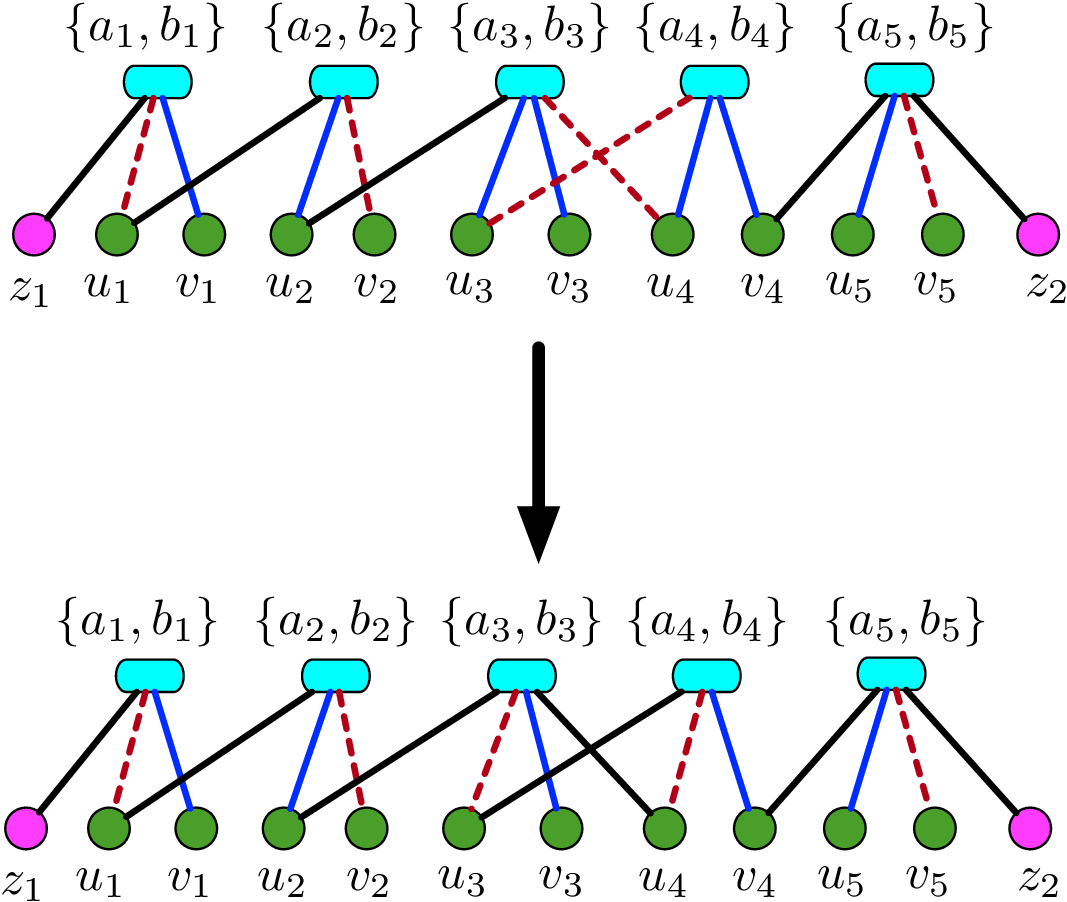}
	\caption{An illustration of Lemma \ref{lemma:vc-respect-2-expansion} proof sketch. Let $\widehat{L_2} = \cup_{i=1}^5 \{\{a_i, b_i\}\}$, $\widehat{R_2} = \big(\cup_{i=1}^5 \{u_i, v_i\}\big) \cup \{z_1, z_2\}$, and $\widehat{R} = \cup_{i=1}^5 \{(\{a_i, b_i\}, u_i), (\{a_i, b_i\}, v_i)\}$.
	Let $\ccP = \{P_1', P_2'\}$ such that $P_1'$ occupies $(\{a_1, b_1\}, u_1)$, $(\{a_3, b_3\}, u_4)$, and $P_2'$ occupies $(\{a_2, b_2\}, v_2), (\{a_4, b_4\}, u_3), (\{a_5, b_5\}, v_5)$.
	After the modification steps, $P_1'$ is converted to $S_1$ and $P_2'$ gets converted to $S_2$.
	Subsequently, $S_1$ occupies $(\{a_1, b_1\}, u_1)$, $(\{a_3, b_3\}, u_3)$ and $S_2$ occupies $(\{a_2, b_2\}, v_2), (\{a_4, b_4\}, u_3), (\{a_5, b_5\}, v_5)$.
	The dotted edges with maroon color represents the edges occupied by the $(A, \ell)$-paths.}
\label{fig:rule-2-path-coversion}
\end{figure}

\begin{lemma}
\label{lemma:char-one-R-respecting}
Let $\ccP$ be a set of $k$ pairwise-vertex-disjoint $(A, \ell)$-paths to $(G, M, A, \ell, k)$ and $\ccP' \subseteq \ccP$ be the set of all $(A,\ell)$-paths that are $\cR$-violating.
Then, for every pair $\{x, y\} \in \widehat{L_2}$ badly occupied by $\ccP'$, the $\cR$-saturated vertices with $\{x, y\}$ cannot be present in any path of $\ccP \setminus \ccP'$.
\end{lemma}

\begin{proof}
Suppose that the premise of the statement holds true and $\ccP' = \{P_1',\ldots, P_r'\}$.
Consider any pair $\{x, y\} \in \widehat{L_2}$ that is badly occupied by $\ccP'$.
Then, $\{x, y\}$ appears in $P_i'$ such that $P_i' \in \ccP'$.
We consider the two vertices $w, z \in \widehat{R_2}$ that are $\cR$-saturated with $\{x, y\}$.
Suppose for the sake of contradiction that $w$ appears in a path $Q \in \ccP \setminus \ccP'$.
Since $N_{H_2}(w) \subseteq \widehat{L_2}$, there exists a pair $\{x', y'\} \in \widehat{L_2}$ such that the path $Q$ occupies the pair $(\{x', y'\}, w)$.
But $\{x, y\}$ is a badly occupied $P_i' \in \ccP'$, and $Q \in \ccP \setminus \ccP'$.
Since $P_i \in \ccP'$, and $P_i'$ is vertex-disjoint with $Q$, observe that $\{x, y\} \neq \{x', y'\}$.
Then, $(\{x', y'\}, w) \in E(\widehat{L_2}, \widehat{R_2}) \setminus \cR$ and $(\{x', y'\}, w)$ is occupied by $Q$ such that $Q \in \ccP \setminus \ccP'$.
It implies that $Q$ is $\cR$-violating.
This contradicts our choice that $Q \notin \ccP'$.
Hence, $w$ cannot be present in any path from $\ccP \setminus \ccP'$.
By similar arguments, we can prove that $z$ cannot be present in any path from $\ccP \setminus \ccP'$.
\end{proof}

Now, we prove the next lemma that is crucial to the safeness of the reduction rule.


\begin{lemma}
\label{lemma:vc-respect-2-expansion}
Let $(G, M, A, \ell, k)$ be a yes-instance.
Then, there exists a feasible solution $\ccP$ such that every $(A, \ell)$-path of $\ccP$ is $\cR$-respecting.
\end{lemma}

\begin{proof}
Let $(G, M, A, \ell, k)$ be a yes-instance and $\ccP^*$ be a set of $k$ vertex-disjoint $(A, \ell)$-paths and for the sake of contradiction, assume that there is a nonempty subset $\ccP' = \{P_1',\ldots,P_{r}'\} \subseteq \ccP^*$ such that every $P_i' \in \ccP'$ is $\cR$-violating.

In the first step, for every $j \in [r]$, we initialize $S_j = P_j'$.
Then, we remove all the badly occupied vertices of $\widehat{R_2}$ from the paths $S_1,\ldots,S_r$.
This removal step makes every $S_j$ into a collections of subpaths.
Consider $u \in \widehat{R_2}$ badly occupied by $\ccP'$.
Then, there is a badly occupied pair $\{x, y\} \in \widehat{L_2}$ such that $(\{x, y\}, u)$ is occupied by one of the paths $P_1',\ldots,P_r'$.
Since $u$ is removed from $S_j$, and $(\{x, y\}, u)$ is occupied by $P_j'$, it follows that after the removal step, $x$ and $y$ are in different subpaths of $S_j$ that {\em need to be connected} by some vertex from $I \setminus A$.
Due to a 2-expansion from $\widehat{L_2}$ onto $\widehat{R_2}$, for every badly occupied pair $\{x, y\} \in \widehat{L_2}$, there are two vertices $w, z \in \widehat{R_2}$ such that $(\{x, y\}, w), (\{x, y\}, z) \in \cR$.
Due to Lemma \ref{lemma:char-one-R-respecting}, neither $w$ nor $z$ appears in any path of $\ccP^* \setminus \ccP'$.
Hence, $w$ or $z$ can appear only in paths that are in $\ccP'$.
If $w$ and $z$ appeared in any path of $\ccP'$, then our removal steps also removed $w$ and $z$.
Consider the corresponding pairs $\{x, y\} \in \widehat{L_2}$ that are badly occupied by $P_i'$ such that
\begin{itemize}
	\item $x, y \in S_i$, and
	\item $x$ and $y$ appear in different different components (that are paths) of $S_i$.
\end{itemize}
Our removal step has ensured that if $w, z \in \widehat{R_2}$ that are $\cR$-saturated with $\{x, y\}$, then $w$ and $z$ are also removed from the paths of $\ccP'$, even if they had existed in any path of $\ccP'$.
Moreover, due to Lemma \ref{lemma:char-one-R-respecting}, $w$ and $z$ cannot appear in any path of $\ccP^* \setminus \ccP'$.
 
So, in this second step, we consider every such pair $\{x, y\} \in \widehat{L_2}$ (one by one) badly occupied by $\ccP'$ and the vertices $x$ and $y$ appear in $S_i$.
We consider the $2$-expansion of $\widehat{L_2}$ onto $\widehat{R_2}$ and the two vertices $w, z \in \widehat{R_2}$ such that $(\{x, y\}, w), (\{x, y\}, z) \in \cR$.
As $w$ and $z$ are not appearing in any other path of $\ccP^* \setminus \ccP'$, hence $w$ can be safely added to $S_j$.
In this step, add $w$ to $S_j$.
We repeat this procedure for every such pair that is badly occupied by any path of $\ccP'$.
A short illustrative example of this conversion procedure is highlighted in Figure \ref{fig:rule-2-path-coversion}.

We claim that this procedure constructs a collection of paths $S_1,\ldots,S_{r}$ each of which is $\cR$-respecting.
Consider any edge $(\{\widehat{x}, \widehat{y}\}, \widehat{w})$ of $E(\widehat{L_2}, \widehat{R_2})$ that is occupied by some path from $S_1,\ldots,S_{r}$.
In the first step of the construction, every edge of $E(\widehat{L_2}, \widehat{R_2})$ badly occupied by $S_1,\ldots,S_r$ had been removed.
Subsequently, for the pair $\{\widehat{x}, \widehat{y}\}$, the vertex $\widehat{w}$ was added to $S_i$ in such a way that $\widehat{w}$ is $\cR$-saturated with $\{\widehat{x}, \widehat{y}\}$.
Hence, $(\{\widehat{x}, \widehat{y}\}, \widehat{w})$ cannot be an edge of $E(\widehat{L_2}, \widehat{R_2})$ that is badly occupied by any path from $S_1,\ldots,S_{r}$.
Hence, every $S_i$ is $\cR$-respecting.

It is not very hard to observe that for every $i \in [r]$, $|S_i| = |P_i'|$.
We set $\ccP = (\ccP^* \setminus \ccP') \cup \{S_1,\ldots,S_r\}$.
Since every path of $\ccP^* \setminus \ccP'$ is $\cR$-respecting and every path of $\{S_1,\ldots,S_r\}$ is $\cR$-respecting, therefore every path of $(\ccP^* \setminus \ccP') \cup \{S_1,\ldots,S_r\}$ is $\cR$-respecting.

Finally, we argue that the paths in $\ccP$ are pairwise vertex-disjoint.
Observe that the paths of $\ccP^* \setminus \ccP'$ are pairwise vertex-disjoint.
Consider $S_i$ and $S_j$ for $i \neq j$.
We look at the construction in the two steps.
In the first step, we only remove vertices.
Removal of the vertices does not violate the disjointness property.
Hence, our first step of construction of $S_1,\ldots,S_r$ ensures the disjointness property.
Next, due to the second step of our construction, for every pair $\{x, y\} \in \widehat{L_2}$ such that $x$ and $y$ are in different connected components of $S_i$, we add $w$ such that $w$ is $\cR$-saturated by $\{x, y\}$.
Since the edges of $\cR$ form a matching due to the properties provided by Proposition \ref{lemma:new-expansion-lemma}, $S_i$ and $S_j$ are pairwise vertex-disjoint.
Finally, we consider a path $Q \in \ccP^* \setminus \ccP'$ and a path $S_i$.
Note that the vertices of $S_i$ that intersect $M$ remain unchanged.
If a pair $\{x, y\} \in \widehat{L_2}$ is occupied by $S_i$, then due to 2-expansion of $\widehat{L_2}$ to $\widehat{R_2}$, there are $w, z \in \widehat{R_2}$ that are $\cR$-saturated by $\{x, y\}$.
Due to Lemma \ref{lemma:char-one-R-respecting}, $w, z \in Q$.
By construction, $w$ (or $z$) appear in $S_i$.
Hence, $S_i$ and $Q$ remain pairwise vertex-disjoint.

This completes the proof that $\ccP$ is a collection of $k$ vertex-disjoint $(A, \ell)$-paths each of which is $\cR$-respecting.
\end{proof}

\begin{lemma}
\label{lemma:bounding-I-minus-A}
Let $(G, M, A, k, \ell)$ be the input instance and $(G', M, A, k, \ell)$ be the instance obtained after applying Reduction Rule \ref{rule:vc-compressing-I-minus-A}.
Then, $(G, M, A, k, \ell)$ be yes-instance if and only if $(G', M, A, k, \ell)$ is a yes-instance.
\end{lemma}

\begin{proof}
The backward direction ($\Leftarrow$) is obvious to see.
Since $G'$ is an induced subgraph of $G$, any set of $k$ pairwise vertex-disjoint $(A, \ell)$-paths in $G'$ is also a collection of pairwise vertex-disjoint $(A, \ell)$-paths in $G$.

We give a proof of the forward direction ($\Rightarrow$).
Since $(G, M, A, k, \ell)$ is a yes-instance, due to Lemma \ref{lemma:vc-respect-2-expansion}, there exists a set $\ccP$ of $k$ vertex-disjoint $A$-paths each of length $\ell$ and each is $\cR$-respecting.
Since every $(A, \ell)$-path appearing in $\ccP$ is a path of $G'$, it follows that $\ccP$ is also a feasible solution, hence a pairwise vertex-disjoint $(A, \ell)$-paths in $(G', M, A, k, \ell)$.
This completes the proof of the forward direction.
 \end{proof}

Using the above mentioned lemmata, we are ready to prove our result of this section that we restate below.

{\vcResult*}

\begin{proof}
Let $(G, M, A, k, \ell)$ be an instance of {\alpp} when $S$ is a vertex cover of $G$ such that $|M| \leq 2{\vc}(G)$.
Let $I = V(G) \setminus M$ and we can assume without loss of generality that $I \setminus A$ has no pendant vertex.
We apply the Reduction Rules \ref{rule:compressing-A_I-part} and \ref{rule:vc-compressing-I-minus-A} exhaustively.
The correctness of Reduction Rule \ref{rule:compressing-A_I-part} follows from Lemma \ref{lemma:vc-compressing-A_I}, and the correctness of Reduction Rule \ref{rule:vc-compressing-I-minus-A} follows from Lemma \ref{lemma:bounding-I-minus-A}.
Once these two reduction rules are not applicable we ensure that $|A \cap I| \leq 2|M \setminus A| \leq 2|M|$ and $|I \setminus A| \leq 2{{|M|}\choose{2}}$.
It means that an irreducible instance $(G, M, A, k, \ell)$ satisfies the property that $G - M$ has at most $2{{|M|}\choose{2}} + 2|M|$ vertices.
As $|M| \leq 2{\vc}(G)$, it follows that {\ALP} parameterized by ${\vc}$, the vertex cover number admits a kernel with $\OO({\vc}^2)$ vertices.
 \end{proof}

%% file: conclusion.tex

\section{Conclusion and Future Research}
\label{sec:conc} 
Our results have extended the works of Belmonte et al. \cite{BelmonteHKKKKLO22} by addressing the parameterized complexity status of  {\ALP} (\alpp) across numerous structural parameters. It was known from Belmonte et al. \cite{BelmonteHKKKKLO22} that \alpp is {\woc} when parameterized by ${\pw} + |A|$. We prove an intractability result for a much larger parameter of $\dtp(G)+|A|$.
Also, the parameterized complexity of \alpp when parameterized by the combined parameter of cliquewidth and $\ell$ was an open question \cite{BelmonteHKKKKLO22}.
While that problem still remains open, we have been successful in making slight progress by obtaining an \fpt algorithm for the problem when parameterized by the combined parameter of ${\cvd}(G)$ and $\ell$. 
Another direction to explore would be to determine the fixed-parameter tractability status of the problem when parameterized by $\cvd(G)$ only. 
It would be interesting to explore if this \fpt result can be generalized to the combined parameter of cograph vertex deletion set size and $\ell$ since cographs are graphs of cliquewidth at most two. 
We believe that the positive results presented in this paper are not optimal and some of those results can be improved with more involved structural analysis.
Therefore, improving the efficiency of our positive results are exciting research direction for future works.

%% file: AL-path.bbl
\begin{thebibliography}{10}

\bibitem{Babu0R22}
Jasine Babu, R.~Krithika, and Deepak Rajendraprasad.
\newblock Packing arc-disjoint 4-cycles in oriented graphs.
\newblock In Anuj Dawar and Venkatesan Guruswami, editors, {\em 42nd {IARCS}
  Annual Conference on Foundations of Software Technology and Theoretical
  Computer Science, {FSTTCS} 2022, December 18-20, 2022, {IIT} Madras, Chennai,
  India}, volume 250 of {\em LIPIcs}, pages 5:1--5:16. Schloss Dagstuhl -
  Leibniz-Zentrum f{\"{u}}r Informatik, 2022.

\bibitem{BandopadhyayBMS24}
Susobhan Bandopadhyay, Aritra Banik, Diptapriyo Majumdar, and Abhishek Sahu.
\newblock Tractability of packing vertex-disjoint a-paths under length
  constraints.
\newblock In Rastislav Kr{\'{a}}lovic and Anton{\'{\i}}n Kucera, editors, {\em
  49th International Symposium on Mathematical Foundations of Computer Science,
  {MFCS} 2024, August 26-30, 2024, Bratislava, Slovakia}, volume 306 of {\em
  LIPIcs}, pages 16:1--16:18. Schloss Dagstuhl - Leibniz-Zentrum f{\"{u}}r
  Informatik, 2024.

\bibitem{BelmonteHKKKKLO22}
R{\'{e}}my Belmonte, Tesshu Hanaka, Masaaki Kanzaki, Masashi Kiyomi, Yasuaki
  Kobayashi, Yusuke Kobayashi, Michael Lampis, Hirotaka Ono, and Yota Otachi.
\newblock {Parameterized Complexity of (A, l)-Path Packing}.
\newblock {\em Algorithmica}, 84(4):871--895, 2022.

\bibitem{boral2016fast}
A.~Boral, M.~Cygan, T.~Kociumaka, and M.~Pilipczuk.
\newblock A fast branching algorithm for cluster vertex deletion.
\newblock {\em Theory of Computing Systems}, 58(2):357--376, 2016.

\bibitem{BruhnU22}
Henning Bruhn and Arthur Ulmer.
\newblock Packing a-paths of length zero modulo four.
\newblock {\em Eur. J. Comb.}, 99:103422, 2022.

\bibitem{ChalermsookCKLM20}
Parinya Chalermsook, Marek Cygan, Guy Kortsarz, Bundit Laekhanukit, Pasin
  Manurangsi, Danupon Nanongkai, and Luca Trevisan.
\newblock From gap-exponential time hypothesis to fixed parameter tractable
  inapproximability: Clique, dominating set, and more.
\newblock {\em {SIAM} J. Comput.}, 49(4):772--810, 2020.

\bibitem{Chudnovsky2008AnAF}
M.~Chudnovsky, William~H. Cunningham, and James~F. Geelen.
\newblock An algorithm for packing non-zero a-paths in group-labelled graphs.
\newblock {\em Combinatorica}, 28:145--161, 2008.

\bibitem{ChudnovskyGGGLS06}
Maria Chudnovsky, Jim Geelen, Bert Gerards, Luis~A. Goddyn, Michael Lohman, and
  Paul~D. Seymour.
\newblock Packing non-zero a-paths in group-labelled graphs.
\newblock {\em Comb.}, 26(5):521--532, 2006.

\bibitem{CyganFKLMPPS15}
Marek Cygan, Fedor~V. Fomin, Lukasz Kowalik, Daniel Lokshtanov, D{\'{a}}niel
  Marx, Marcin Pilipczuk, Michal Pilipczuk, and Saket Saurabh.
\newblock {\em Parameterized Algorithms}.
\newblock Springer, 2015.

\bibitem{DellHMTW14}
Holger Dell, Thore Husfeldt, D{\'{a}}niel Marx, Nina Taslaman, and Martin
  Wahlen.
\newblock Exponential time complexity of the permanent and the tutte
  polynomial.
\newblock {\em {ACM} Trans. Algorithms}, 10(4):21:1--21:32, 2014.

\bibitem{Diestel-Book}
Reinhard Diestel.
\newblock {\em Graph Theory, 4th Edition}, volume 173 of {\em Graduate texts in
  mathematics}.
\newblock Springer, 2012.

\bibitem{DowneyF13}
Rodney~G. Downey and Michael~R. Fellows.
\newblock {\em Fundamentals of Parameterized Complexity}.
\newblock Texts in Computer Science. Springer, 2013.

\bibitem{FellowsHRV09}
Michael~R. Fellows, Danny Hermelin, Frances~A. Rosamond, and St{\'{e}}phane
  Vialette.
\newblock On the parameterized complexity of multiple-interval graph problems.
\newblock {\em Theor. Comput. Sci.}, 410(1):53--61, 2009.

\bibitem{FominLLSTZ19}
Fedor~V. Fomin, Tien{-}Nam Le, Daniel Lokshtanov, Saket Saurabh, St{\'{e}}phan
  Thomass{\'{e}}, and Meirav Zehavi.
\newblock Subquadratic kernels for implicit 3-hitting set and 3-set packing
  problems.
\newblock {\em {ACM} Trans. Algorithms}, 15(1):13:1--13:44, 2019.

\bibitem{MengerTheoremLength}
Petr~A. Golovach and Dimitrios~M. Thilikos.
\newblock Paths of bounded length and their cuts: Parameterized complexity and
  algorithms.
\newblock {\em Discret. Optim.}, 8:72--86, 2011.

\bibitem{HiraiP14}
Hiroshi Hirai and Gyula Pap.
\newblock Tree metrics and edge-disjoint s-paths.
\newblock {\em Math. Program.}, 147(1-2):81--123, 2014.

\bibitem{JacobMR23}
Ashwin Jacob, Diptapriyo Majumdar, and Venkatesh Raman.
\newblock Expansion lemma - variations and applications to polynomial-time
  preprocessing.
\newblock {\em Algorithms}, 16(3):144, 2023.

\bibitem{Mulmuley1987}
Ketan Mulmuley, Umesh~V. Vazirani, and Vijay~V. Vazirani.
\newblock Matching is as easy as matrix inversion.
\newblock In {\em Proceedings of the Nineteenth Annual ACM Symposium on Theory
  of Computing}, pages 345--354. Association for Computing Machinery, 1987.

\bibitem{Pap08}
Gyula Pap.
\newblock Packing non-returning a-paths algorithmically.
\newblock {\em Discret. Math.}, 308(8):1472--1488, 2008.

\end{thebibliography}
